\documentclass[review]{elsarticle}

\usepackage{lineno,hyperref}
\modulolinenumbers[5]

\usepackage[T1]{fontenc}
\usepackage[latin9]{inputenc}
\usepackage{color}
\usepackage{array}
\usepackage{float}
\usepackage{mathrsfs}
\usepackage{mathtools}
\usepackage{amsthm}
\usepackage{amstext}
\usepackage{amssymb}
\usepackage{stmaryrd}
\usepackage{graphicx}
\usepackage{pgfplots}
\usepackage{makecell}

\allowdisplaybreaks

\makeatletter
\numberwithin{equation}{section}  

\theoremstyle{plain}
\newtheorem{theorem}{\protect\theoremname}
\theoremstyle{plain}

\theoremstyle{plain}
\newtheorem{corollary}{\protect\corollaryname}
\theoremstyle{plain}
\newtheorem{lemma}{\protect\lemmaname}
\theoremstyle{definition}
\newtheorem{example}{\protect\examplename}
\theoremstyle{definition}
\newtheorem{definition}{\protect\definitionname}
\theoremstyle{definition}
\newtheorem{remark}{\protect\remarkname}
\theoremstyle{definition}

\AtBeginDocument{
	
}

\makeatother

  \providecommand{\corollaryname}{Corollary}
  \providecommand{\examplename}{Example}
  \providecommand{\lemmaname}{Lemma}
  \providecommand{\propositionname}{Proposition}
  \providecommand{\theoremname}{Theorem}
  \providecommand{\definitionname}{Definition}
  \providecommand{\remarkname}{Remark}
  \providecommand{\acknowledgementsname}{Acknowledgements}

\journal{Finite Fields and Their Applications}

\newcommand{\Bk}[1]{{\Big( {#1} \Big) }}
\newcommand{\bk}[1]{{\big( {#1} \big) }}

\begin{document}

\begin{frontmatter}

\title{Linear Codes Constructed From Two Weakly Regular Plateaued Functions with Index $ (p-1)/2 $ \tnoteref{mytitlenote}
	}

\author[mymainaddress]{Shudi Yang\corref{mycorrespondingauthor}}
\cortext[mycorrespondingauthor]{Corresponding author}
\ead{yangshudi@qfnu.edu.cn}

\author[mysecondaryaddress,mymainaddress]{Tonghui Zhang}
\ead{zhangthvvs@126.com}

\author[mythirdaddress]{Zheng-An Yao}
\ead{mcsyao@mail.sysu.edu.cn}

\address[mymainaddress]{School of Mathematical
	Sciences, Qufu Normal University, Shandong 273165, P.R.China}

\address[mysecondaryaddress]{School of Mathematics and
	Statistics,Fujian Normal University, Fujian 350117, P.R.China}
	
\address[mythirdaddress]{School of Mathematics, Sun Yat-Sen University, Guangzhou 510275,  P.R.China}

\begin{abstract}
Linear codes are the most important family of codes in cryptography and coding theory. Some codes have only a few weights and are widely used in many areas, such as authentication codes, secret sharing schemes and strongly regular graphs. By setting $ p\equiv 1 \pmod 4 $, we construct an infinite family of linear codes using two distinct weakly regular unbalanced (and balanced) plateaued functions with index $ (p-1)/2 $. Their weight distributions are completely determined by applying exponential sums and Walsh transform. As a result, most of our constructed codes have a few nonzero weights and are minimal.

\end{abstract}


\end{frontmatter}


\section{Introduction}\label{sec:intro}

Throughout the paper, we always let $p$ be an odd prime. We will use the symbol $ \mathbb{F}_{p} $ to denote the finite field with $p$ elements. A linear code $C$ over $\mathbb{F}_{p}$ with length $ n $, dimension $ k $ and minimum distance $ d $ is said to have parameters $ [n,k,d] $, which means that $ C $ is a $k$-dimensional linear subspace of $\mathbb{F}_{p}^n$
 with Hamming distance $d$. The code $ C $ is called
 projective if its dual code has minimum distance larger than $ 2 $.
 The Hamming weight of a codeword $\textbf{c}$, denoted by $\texttt{wt}(\textbf{c})$, is defined as the number of nonzero entries in $ \textbf{c} $. 
 Let $A_w = \# \{\textbf{c} \in C : \texttt{wt}(\textbf{c})=w \}$ for $ 0\leqslant w \leqslant n $. Then the sequence $\left(A_0,A_1,A_2,\dots,A_n \right)$ stands for the weight distribution of $ C $, where $ A_0=1 $. The code $ C $ is called $t$-weight if the number of nonzero $A_w$ for $1\leqslant w \leqslant n $ equals $t$. The weight distribution is of vital importance since it contains the information of computing the error probability of error detection and correction. In recent decades, a large number of linear codes have been investigated, most of which have a few weights and good parameters \cite{booleanchen2021,cao2021,heng2018,heng2021,kong2019,lu2020,mes2020,Sinak2021,yang2021,zhang2022,zheng2021,bent2022}. The construction of linear codes is usually based on different functions, such as, Boolean functions \cite{booleanchen2021}, bent functions \cite{bent2016tang,bent2022}, square functions \cite{squaretang2016} and weakly regular plateaued functions \cite{cao2021,mes2020,Sinak2021}.
 
 Let us introduce an efficient way to construct linear codes, which was proposed by Ding \emph{et al.} \cite{ding2007}. Let $ q=p^m $ for an integer $ m $, and $D$ be a subset of $ \mathbb{F}_q^*$ of size $ n $. Define
 \begin{align*}
 C_D=\left\{\textbf{c}(a)=\left(\textup{Tr}(ax) \right)_{x \in D}:a\in\mathbb{F}_q \right \},
 \end{align*}
 where $\textup{Tr}$ is the absolute trace function. It can be checked that $ C_D $ is a linear code of length $ n $. The set $ D $ is called the defining set of $ C_D $. This approach is soon generalized by Li \emph{et al.} \cite{li2017}, who defined a class of codes by
 \begin{align}\label{eq:code}
 C_D=\left\{\textbf{c}(a,b)=\left(\textup{Tr}(ax+by)\right)_{(x,y)\in D}: a,b \in \mathbb{F}_q\right \},
 \end{align}
 where the defining set $ D $ is a subset of $ \mathbb{F}_q^2 \backslash \{(0,0)\} $ of size $ n $. Based on this method, Wu \emph{et al.} \cite{wu2020} offered a new approach to linear codes using the defining set 
 \begin{align}\label{def:D}
 D= \left \{ (x,y) \in \mathbb{F}_q^2\backslash \{(0,0)\}:f(x)+g(y) =0 \right \},
 \end{align}
 where $f$ and $g$ are weakly regular bent functions from $ \mathbb{F}_q $ to $ \mathbb{F}_p $.
  Later, Cheng \emph{et al.} in \cite{cao2021} introduced several linear codes $ C_D $ of \eqref{eq:code} with a few weights by considering $f$ and $g$ to be weakly regular unbalanced $s$-plateaued functions in the defining set \eqref{def:D}, where $ 0\leqslant s \leqslant m $. In 2022, S{\i}nak \cite{Sinak2022} went deeper by choosing weakly regular unbalanced and balanced $s_f$-plateaued function $f$ and $s_g $-plateaued function $g$ in the defining set \eqref{def:D}, where $ 0\leqslant s_f,\,s_g \leqslant m $. All of them studied the indexes of functions $ f $ and $ g $ among the set $ \{2,p-1\} $, that is, $  l_f,l_g \in \{2,p-1\} $.
  
  Along the research line studied in \cite{cao2021,Sinak2022,wu2020}, we further consider the index of $ \frac{p-1}{2} $, where $ p \equiv 1 \pmod 4 $. Now the defining set is denoted by
 \begin{align}\label{def:D-fg}
 D_{f,g}= \left \{ (x,y) \in \mathbb{F}_q^2 \backslash \{(0,0)\}:f(x)+g(y) =0 \right \},  
 \end{align}
 where $f$ and $g$ are weakly regular unbalanced and balanced $s$-plateaued and $t$-plateaued functions, respectively, for $ 0\leqslant s,\,t \leqslant m $. For clarity, we only concentrate on the case $ l_g=\frac{p-1}{2} $ and $ l_f \in \{2,\frac{p-1}{2},p-1\} $, where $ p \equiv 1 \pmod 4 $. 
 In this paper, we consider the constructed codes $ C_{D_{f,g}} $ of \eqref{eq:code} and \eqref{def:D-fg}. In details, we will determine their parameters and their weight distributions using Walsh transform.
Their punctured codes are also determined. As we will show later, they are projective, and some of them are optimal since they meet the Griesmer bound.
 
 The rest of this paper is arranged as follows. We firstly present in
 Section \ref{sec:math} an introduction to the mathematical foundations, including cyclotomic fields and weakly regular plateaued functions.  
 Section 3 gives necessary results for our computation. Our main results are proposed in Section 4, where we study the weight distributions and the parameters of our constructed codes and their punctured ones. 
 Section 5 shows the minimality and applications of these codes.
 Finally, we conclude the whole paper in Section 6.

 \section{Mathematical background}\label{sec:math}
 
 In this section, we give a brief exposition of exponential sums, cyclotomic fields, cyclotomic classes and weakly regular plateaued functions. 
 First of all, we set up notation and terminology.
Let $ q=p^m $ for an integer $ m \geqslant 2 $.
The set of square (resp. non-square) elements in $\mathbb{F}_{p}^*$ is denoted by $ S_q$ (resp. $ N_{sq}$). Let $ \eta $ be the quadratic character of $\mathbb{F}_{p} $, that is,
\begin{align*}
\eta(x)=\begin{cases}
0, & \textup{if } x =0,\\
1, & \textup{if } x \in S_q,\\
-1, & \textup{if } x \in N_{sq}.
\end{cases}
\end{align*}

 \subsection{Cyclotomic classes and cyclotomic fields}
 
 Let $ \theta $ be a fixed primitive element of $ \mathbb{F}_{q} $ and $ N \geqslant 2 $ be a divisor of $ q-1 $. The
 cyclotomic classes of order $  N $ in $ \mathbb{F}_{q} $ are defined by $ C_i^{(N,q)} = \theta^i\langle \theta^N  \rangle  $
  for $ i = 0, 1, \ldots ,N - 1 $,
 where $ \langle \theta^N  \rangle $ stands for the subgroup of $ \mathbb{F}_{q}^* $ generated by $ \theta^N $. Obviously, we have $ \mathbb{F}_{q}^*=  \bigcup_{i=0}^{N-1} C_i^{(N,q)} $ and every cyclotomic class has the same number of elements, that is $ \# C_i^{(N,q)} = \frac{q-1}{N} $. 
 
The $p$-th cyclotomic field is denoted by $K=\mathbb{Q}(\zeta_p)$, where $ \zeta_p = \exp\bk{\frac{2\pi \sqrt{-1}}{p}} $. The following lemma enunciates useful properties of this field.

\begin{lemma}\label{lem:cyclo}(\cite{ire2013})
	The following assertions hold for $K=\mathbb{Q}(\zeta_p)$. 
	
$ (1) $ The ring of integers in $ K $ is $ \mathbb{Z}[\zeta_p]$, 
and $\{ \zeta_p^i:1\leqslant i \leqslant p-1\} $ is an integer basis of $ \mathbb{Z}[\zeta_p]$. 

$ (2) $ The field extension  $K/\mathbb{Q}$ is Galois of degree $p-1$, and the Galois group $\textup{Gal}(K/\mathbb{Q})=\{\sigma_z: z \in \mathbb{F}_p^*\}$, where the automorphism $ \sigma_z $ of $ K $ is defined by $ \sigma_z(\zeta_p)=\zeta_p^z $.

$ (3) $ The cyclotomic field $ K $ has a unique quadratic subfield $\mathbb{Q}(\sqrt{p^*})$, where $ p^*= \eta(-1)p $. For $ z \in \mathbb{F}_p^* $,  $\sigma_z(\sqrt{p^*})= \eta(z)\sqrt{p^*} $. 
\end{lemma}
It is easily proved from Lemma \ref{lem:cyclo} that $\sigma_z(\sqrt{p^*}^m)=\eta^m(z)\sqrt{p^*}^m$ and $\sigma_z(\zeta_p^a)=\zeta_p^{za}$ for all $a \in \mathbb{F}_p$.

 \subsection{Exponential sums}
 In this subsection, we briefly sketch the concept of exponential sums.
 Let $\eta_m $ denote the quadratic character of $\mathbb{F}_q $, where $ q=p^m $. 
 The quadratic Gauss sum over $\mathbb{F}_q $ is defined by
 \begin{align*}
 G(\eta_m )=\sum_{x\in\mathbb{F}_q^*}\eta_m (x)\chi_1(x),
 \end{align*}
where $ \chi_1(x)=\zeta_p^{\text{Tr}(x)}$ is the canonical additive character  of $\mathbb{F}_q $, and $\text{Tr}$ is the absolute trace function. From Theorem 5.15 in \cite{lidl1997},	
 $ G(\eta_m)=(-1)^{m-1}\sqrt{p^*}^m $ and
 $  G(\eta )= \sqrt{p^*} $.
 
 For $ n \in \mathbb{N} $ and $ a \in \mathbb{F}_{q}^* $, the Jacobsthal sum is defined by 
 \begin{align*}
 H_n(a) = \sum_{x\in \mathbb{F}_{q}} \eta_m(x^{n+1}+ax) =  \sum_{x\in \mathbb{F}_{q}} \eta_m (x)\eta_m(x^{n}+a).
 \end{align*}
 Define 
  \begin{align*}
 I_n(a) = \sum_{x\in \mathbb{F}_{q}} \eta_m(x^{n}+a).
  \end{align*}
It is a companion sum related to Jacobsthal sums because $ I_{2n}(a)=I_n(a) + H_n(a) $, which is due to Theorem 5.50 in \cite{lidl1997}. We can evaluate easily that $ I_1(a)=0 $ and $ I_2(a)=-1 $ for all $ a \in \mathbb{F}_{q}^* $.
In general, the sums $  I_n(a) $ can be described in terms of Jacobi sums.

 \begin{lemma}[Theorem 5.51, \cite{lidl1997}]\label{lm:companion}
 For all $ a \in \mathbb{F}_{q}^* $ and $ n \in \mathbb{N} $, we have
  	\begin{equation*}
 I_n(a) = \eta_m(a)\sum_{j=1}^{d-1} \lambda^j(-a)J(\lambda^j,\eta_m),
 	\end{equation*} 
 where $ \lambda $ is a multiplicative character of $ \mathbb{F}_{q} $ of order $ d=\gcd(n,q-1) $, and $ J(\lambda^j,\eta_m) $ is a Jacobi sum in $\mathbb{F}_q $. 
  \end{lemma}
 
  \begin{lemma}[Theorem 5.33, \cite{lidl1997}]\label{lem:expo sum}
  	Let $ q =p^m $ be odd and $ f(x)=a_2 x^2 +a_1x +a_0\in \mathbb{F}_{q}[x] $ with
  	$ a_2\neq 0 $.	Then
  	\begin{equation*}
  	\sum_{x\in
  		\mathbb{F}_{q}}\chi_1(f(x))= \chi_1 (a_0- a_1^2(4a_2)^{-1})\eta_m (a_2)G(\eta_m).
  	\end{equation*} 
  \end{lemma}

 \subsection{Weakly regular plateaued functions}
 
We now introduce weakly regular plateaued functions and review some basic facts about them. 
 Let $f:\mathbb{F}_q  \rightarrow \mathbb{F}_p$ be a $ p $-ary function. For $ \beta \in \mathbb{F}_q  $, the Walsh transform of $ f $ is defined as a complex-valued function $ \widehat{\chi}_f $ on $ \mathbb{F}_q  $, 
 \begin{align*}
 \widehat{\chi}_f(\beta)=\sum_{x\in \mathbb{F}_{q}}\zeta_p^{f(x)-\textup{Tr}(\beta x)}.
 \end{align*}
A $ p $-ary function $ f $ is called to be balanced if it satisfies $ \widehat{\chi}_f(0)=0 $; otherwise, it is called unbalanced over $ \mathbb{F}_q $.
 
 
As a natural extension of bent functions, Zheng \emph{et al.} firstly set up the concept of plateaued functions in characteristic $ 2 $ in \cite{zhang1999}, and later Mesnager \cite{Mes2014} gave a general version of any characteristic $ p $. Several years ago, Mesnager \emph{et al.} presented the notion of (non)-weakly regular plateaued functions in their work \cite{Mes2019}. We follow the notation used in \cite{Mes2019}. For each $\beta \in \mathbb{F}_{q}$, a function $ f $ is called $ s $-plateaued if $|\widehat{\chi}_f(\beta)|^2 \in \{0,  p^{m+s}\}$, where $0 \leqslant s \leqslant m $. It is worth noting that every bent function is $0$-plateaued. The Walsh support of function $ f $, denoted by $ \mathcal{S}_f $, is defined by
 \begin{align*}
 \mathcal{S}_f =\{\beta \in \mathbb{F}_{q}:|\widehat{\chi}_f(\beta)|^2=p^{m+s}\}.
 \end{align*}
Applying the Parseval identity, one gets the absolute Walsh distribution of plateaued functions.

 \begin{lemma}\label{lem:car}(Lemma 1, \cite{Mes2014})
 	Let $f $ be an $s$-plateaued function. Then for $ \beta \in \mathbb{F}_{q}$, $|\widehat{\chi}_f(\beta)|^2 $ takes the value $p^{m+s}$ for $p^{m-s}$ times and the value $0$ for $p^{m}-p^{m-s}$ times.
 \end{lemma}
 
 From Lemma \ref{lem:car}, the cardinality of $ \mathcal{S}_f $ is given by $\# \mathcal{S}_f=p^{m-s}$.

 \begin{definition}\label{defa}(\cite{Mes2019})
 	Let $ f $ be an $s$-plateaued function, where $ 0 \leqslant s \leqslant m $. Then, $f$ is called weakly regular $s$-plateaued if there exists a complex number $ u $ having unit magnitude such that
 	\begin{align*}
 	\widehat{\chi}_f(\beta) \in \{0,up^{\frac{m+s}{2}}\zeta_p^{g(\beta)}\}
 	\end{align*}
 	for all $\beta \in \mathbb{F}_{q}$, where $g$ is a $p$-ary function over $\mathbb{F}_q$ satisfying $ g(\beta)=0 $ for all $\beta \in \mathbb{F}_q \setminus \mathcal{S}_f $. Otherwise, if $ u $ depends on $ \beta $ with $ |u|=1 $, then $f$ is called non-weakly regular $s$-plateaued. Particularly, a weakly regular plateaued function $f$ is said to be regular plateaued if $u=1$.
 \end{definition}
 \begin{lemma}\label{lem:Walsh-f}(Lemma 5, \cite{Mes2019})
 	Let $\beta \in \mathbb{F}_{q}$ and  $f $ be a weakly regular $s$-plateaued function. For every  $\beta \in \mathcal{S}_f $ we have	
 	\begin{align*}
 	\widehat{\chi}_f(\beta)=\varepsilon_f \sqrt{p^*}^{m+s}\zeta_p^{f^{\star}(\beta)},
 	\end{align*}
 	where $\varepsilon_f \in \{ \pm 1 \} $ is the sign of $\widehat{\chi}_f $ and $f^{\star}$ is a $p$-ary function over $\mathbb{F}_q $ with $ f^{\star}(\beta)=0 $ for all $\beta \in \mathbb{F}_q \setminus \mathcal{S}_f $. We call $ f^{\star} $ the dual function of $ f $.	
 \end{lemma}

 We now introduce two non-trivial subclasses of weakly regular plateaued functions. Let $f $ be a weakly regular unbalanced (resp. balanced) $s$-plateaued function with $ 0 \leqslant s \leqslant m $. We denote by \rm{WRP} (resp. \rm{WRPB}) the subclass of the unbalanced (resp. balanced) functions $ f $ that meet the following homogeneous conditions simultaneously:
 \begin{enumerate}
 	\item $f(0)=0$;
 	\item There exists a positive integer $h$, such that $ 2\mid h $, $ \gcd(h-1, p-1)=1$ and $f(zx)=z^hf(x)$ for any $z \in \mathbb{F}_p^*$ and $x \in \mathbb{F}_q$.
 \end{enumerate}
 \begin{remark}
 	For every $f \in \rm{WRP}$ (resp. $f \in \rm{WRPB}$), we have $0 \in \mathcal{S}_f$ (resp. $0 \notin \mathcal{S}_f$).
 \end{remark}

The following lemmas, due to \cite{mes2020} and \cite{Sinak2022}, play a significant role in calculating the parameters of our
constructed codes.

 \begin{lemma}\label{lemc}(Lemma 6, \cite{mes2020})
 	Let $f \in \rm{WRP}$ or $f \in \rm{WRPB}$ with $\widehat{\chi}_f(\beta)=\varepsilon_f \sqrt{p^*}^{m+s}\zeta_p^{f^{\star}(\beta)}$ for every $\beta \in  \mathcal{S}_f$. Then, for every $z \in \mathbb{F}_p^*$, $z\beta \in \mathcal{S}_f$ if $\beta \in \mathcal{S}_f$, and otherwise, $z\beta \in \mathbb{F}_q  \backslash \mathcal{S}_f $.
 \end{lemma}
 
 \begin{lemma}\label{lem:dual}(Propositions 2 and 3, \cite{mes2020})
 	Let $f \in \rm{WRP}$ or $f \in \rm{WRPB}$ with $\widehat{\chi}_f(\beta)=\varepsilon_f \sqrt{p^*}^{m+s}\zeta_p^{f^{\star}(\beta)}$ for every $\beta \in  \mathcal{S}_f$. Then $f^{\star}(0)=0$ and $f^{\star}(z \beta)=z^{l_f} f^{\star}( \beta)$ for all $z \in \mathbb{F}_p^*$ and $\beta \in  \mathcal{S}_f$, where $l_f$ is an even positive integer with $\gcd(l_f-1, p-1)=1$. We call $l_f$ the index of $ f $. 
 \end{lemma}

 
 \begin{lemma}\label{lem:N_f}(Lemma 10, \cite{mes2020})
 	Let $f \in \rm{WRP}$ or $f \in \rm{WRPB}$ with $\widehat{\chi}_f(\beta)=\varepsilon_f \sqrt{p^*}^{m+s}\zeta_p^{f^{\star}(\beta)}$ for every $\beta \in \mathcal{S}_f$. For $c \in \mathbb{F}_p$, define
 	\begin{align*}
 	\mathcal{N}_{f}(c)=\# \{\beta \in \mathcal{S}_f: f^{\star}(\beta)=c\}.
 	\end{align*}
 	When $ m-s $ is even, we have
 	\begin{align*}
 	\mathcal{N}_{f}(c)=\begin{cases}
 	p^{m-s-1}+(p-1)\eta^{m+1}(-1) \varepsilon_f \sqrt{p^*} ^{m-s-2}, &\textup{if }  c=0,\\
 	p^{m-s-1}-\eta^{m+1}(-1)\varepsilon_f \sqrt{p^*}^{m-s-2}, & \textup{if }    c \in \mathbb{F}_p^*.
 	\end{cases}
 	\end{align*}
 	Otherwise,
 	\begin{align*}
 	\mathcal{N}_{f}(c)=\begin{cases}
 	p^{m-s-1}, & \textup{if }  c=0,\\
 	p^{m-s-1}+ \eta(c)\eta^{m}(-1) \varepsilon_f \sqrt{p^*}^{m-s-1}, & \textup{if }    c  \in \mathbb{F}_p^*. 
 	\end{cases}
 	\end{align*}
 	
 \end{lemma}

  \begin{lemma}(Lemma 3.12, \cite{Sinak2022})\label{lem:T}
  Let $f,\,g \in \rm{WRP}$ or $f,\,g \in \rm{WRPB}$, with $\widehat{\chi}_f(\alpha)=\varepsilon_f\sqrt{p^*}^{m+s}\zeta_p^{f^{\star}(\alpha)}$ and $\widehat{\chi}_g(\beta)=\varepsilon_g\sqrt{p^*}^{m+t}\zeta_p^{g^{\star}(\beta)}$ for $\alpha \in \mathcal{S}_f$ and $\beta \in \mathcal{S}_g$, respectively. Define
  	\begin{align*}
  	&\mathcal{T}(0) = \# \{(a,  b) \in  \mathcal{S}_f \times \mathcal{S}_g  :  f^{\star}(a)+g^{\star}(b)=0\} ,\\
  	&\mathcal{T}(c) = \# \{(a,  b) \in  \mathcal{S}_f \times \mathcal{S}_g :  f^{\star}(a)+g^{\star}(b)=c\}, \textup{ where } c \in \mathbb{F}_p^* .
  	\end{align*}
  	Then we have
  	\begin{align*}
  	&\mathcal{T}(0)=\begin{cases}
  	p^{2m-s-t-1}+(p-1)p^{-1}  \varepsilon_f \varepsilon_g\sqrt{p^*}^{2m-s-t}, & \textup{if }   s+t  \textup{ is even},\\
  	p^{2m-s-t-1},&  \textup{if }    s+t  \textup{ is odd},
  	\end{cases}\\
  	&\mathcal{T}(c)=\begin{cases}
  	p^{2m-s-t-1}-p^{-1} \varepsilon_f \varepsilon_g\sqrt{p^*}^{2m-s-t}, & \textup{if }   s+t  \textup{ is even},\\
  	p^{2m-s-t-1}+  \eta(c) \varepsilon_f \varepsilon_g\sqrt{p^*}^{2m-s-t-1}, & \textup{if }   s+t  \textup{ is odd}.
  	\end{cases}
  	\end{align*}
  \end{lemma}
 
  \begin{lemma}\label{lem:length}(Lemma 3.7, \cite{Sinak2022})
  Let $n=\#D_{f,g}$, where $ D_{f,g} $ is defined by \eqref{def:D-fg} with $ f$ and $g $ be as in Lemma \ref{lem:T}. If $ f,\,g \in \rm{WRPB} $, then $ n=p^{2m-1}-1 $. If $ f,\,g \in \rm{WRP} $, then 
  	\begin{align*}
  	n =\begin{cases}
  	p^{2m-1}-1, &\textup{if }  2\nmid s+t,\\
  	p^{2m-1}-1+(p-1)p^{-1} \varepsilon_f \varepsilon_g  \sqrt{p^*}^{2m+s+t}, & \textup{if }  2\mid s+t  .
  	\end{cases}
  	\end{align*}
  \end{lemma}

 \section{Auxiliary results}
To get the frequencies of codewords in the constructed codes, we will need several lemmas which are depicted and proved in the sequel.
\begin{lemma}\label{lem:eta}
Let $ p\equiv 1 \pmod 2 $. For the quadratic character $ \eta $ over $\mathbb{F}_{p}$, we have
\begin{align*}
\sum_{u\in S_q} \sum_{\substack{v\in S_q\\ v \ne \pm u} } 
\eta(u+v) & = -\frac{p-1}{2} (\eta(2)+1),\\
\sum_{u\in N_{sq}} \sum_{\substack{v\in N_{sq}\\ v \ne \pm u} } 
\eta(u+v) & = \frac{p-1}{2} (\eta(2)+1).
\end{align*}	
\end{lemma}
\begin{proof}
Notation that $ - 1 \in S_q$ when $ p\equiv 1 \pmod 4 $, and $ - 1 \in N_{sq}$ when $ p\equiv 3 \pmod 4 $. It follows that
\begin{align*}
\sum_{u\in S_q} \sum_{\substack{v\in S_q\\ v \ne \pm u} } 
\eta(u+v) & =  \sum_{u\in S_q} \eta(u)\sum_{\substack{v\in S_q\\ v \ne \pm u} } 
\eta(1+\frac{v}{u})\\
& =  \sum_{u\in S_q}  \sum_{\substack{v\in S_q\\ v \ne \pm 1} } 
\eta(1+v)\\
& =  \frac{p-1}{2}  \Bk{  \sum_{ v\in S_q } \eta(1+v) -\eta(2) }\\
& =  \frac{p-1}{2}  \Bk{ \frac{1}{2} \sum_{ x\in \mathbb{F}_p  } \eta(1+x^2) -\frac{1}{2}  -\eta(2) } \\
& =  \frac{p-1}{2}  \Bk{ \frac{1}{2} I_2(1) -\frac{1}{2}  -\eta(2) } .
\end{align*}
The first assertion then follows from the fact that $ I_2(1)=-1 $. The second one is analogously proved and is omitted here.	
\end{proof}

\begin{lemma}\label{lem:S_q-N_{sq}} 
Let $ p\equiv 1 \pmod 4 $ and $f,\,g \in \rm{WRP}$ or $f,\,g \in \rm{WRPB}$, with $\widehat{\chi}_f(\alpha)=\varepsilon_f\sqrt{p^*}^{m+s}\zeta_p^{f^{\star}(\alpha)}$ and $\widehat{\chi}_g(\beta)=\varepsilon_g\sqrt{p^*}^{m+t}\zeta_p^{g^{\star}(\beta)}$ for every $\alpha \in \mathcal{S}_f $ and every $\beta \in \mathcal{S}_g$, respectively. Suppose that $ s+t $ is odd. Write
	\begin{align*}
	B_{S_q}& = \# \{(a,  b) \in  \mathcal{S}_f \times \mathcal{S}_g  :  f^{\star}(a)+g^{\star}(b)\in S_q,f^{\star}(a)-g^{\star}(b)\in S_q\} ,\\
	B_{N_{sq}}& = \# \{(a,  b) \in  \mathcal{S}_f \times \mathcal{S}_g  :  f^{\star}(a)+g^{\star}(b)\in N_{sq},f^{\star}(a)-g^{\star}(b)\in N_{sq}\} .
	\end{align*}
Then if $ 2\nmid m-s $ and  $ 2\mid m-t $, we have 
	\begin{align*}
	B_{S_q}& = \frac{p-1}{2}\sqrt{p}^{2m-s-t-3}  \Big( \frac{p-1}{2} \sqrt{p}^{2m-s-t-1} - \eta(2)\varepsilon_f  \sqrt{p}^{m-t}\\
		& \phantom{==} + \frac{p+1}{2} \varepsilon_g \sqrt{p}^{m-s-1} +(\eta(2)+p)\varepsilon_f \varepsilon_g 	\Big),\\
	B_{N_{sq}}& = \frac{p-1}{2}\sqrt{p}^{2m-s-t-3}  \Big( \frac{p-1}{2} \sqrt{p}^{2m-s-t-1} + \eta(2)\varepsilon_f  \sqrt{p}^{m-t}\\
     	& \phantom{==} + \frac{p+1}{2} \varepsilon_g \sqrt{p}^{m-s-1} - (\eta(2)+p)\varepsilon_f \varepsilon_g	\Big).
	\end{align*}
Otherwise if $ 2\mid m-s $ and  $ 2\nmid m-t $, we have  
\begin{align*}
B_{S_q}& = \frac{p-1}{2}\sqrt{p}^{2m-s-t-3}  \Big( \frac{p-1}{2} \sqrt{p}^{2m-s-t-1} - \eta(2)\varepsilon_g  \sqrt{p}^{m-s}\\
& \phantom{==} + \frac{p+1}{2} \varepsilon_f \sqrt{p}^{m-t-1} +(\eta(2)+p)\varepsilon_f \varepsilon_g 	\Big),\\
B_{N_{sq}}& = \frac{p-1}{2}\sqrt{p}^{2m-s-t-3}  \Big( \frac{p-1}{2} \sqrt{p}^{2m-s-t-1} + \eta(2)\varepsilon_g  \sqrt{p}^{m-s}\\
& \phantom{==} + \frac{p+1}{2} \varepsilon_f \sqrt{p}^{m-t-1} - (\eta(2)+p)\varepsilon_f \varepsilon_g	\Big).
\end{align*}
\end{lemma}
\begin{proof}
We only calculate $ B_{S_q} $ and omit the other. Now suppose that $ 2\nmid m-s $ and  $ 2\mid m-t $. Let $  f^{\star}(a)+g^{\star}(b)=u$, $ f^{\star}(a)-g^{\star}(b)=v $, where $ u,v\in \mathbb{F}_p^* $. Then we have $  f^{\star}(a)=\frac{u+v}{2} $,  $ g^{\star}(b)= \frac{u-v}{2} $ and so
	\begin{align*} 
	B_{S_q} = \sum_{u\in S_q} \sum_{v\in S_q} \mathcal{N}_{f}(\frac{u+v}{2})	\mathcal{N}_{g}(\frac{u-v}{2}),
	\end{align*}
where $ \mathcal{N}_{f} $ and $ \mathcal{N}_{g} $ are computed in Lemma \ref{lem:N_f}. It follows that
	\begin{align*} 
	B_{S_q} & = \sum_{u\in S_q} \mathcal{N}_{f}(u) \mathcal{N}_{g}(0) + \sum_{u\in S_q} \mathcal{N}_{f}(0) \mathcal{N}_{g}(u) + S,
	\end{align*}
where
\begin{equation}\label{eq:S1}
S=\sum_{u\in S_q} \sum_{\substack{v\in S_q\\ v \ne \pm u} } \mathcal{N}_{f}(\frac{u+v}{2})	\mathcal{N}_{g}(\frac{u-v}{2}).
\end{equation}
Observe that $ \frac{u-v}{2} \ne 0 $ in \eqref{eq:S1}. If we write $ c=\frac{u-v}{2} \ne 0  $, then from Lemma \ref{lem:N_f},
\begin{align*} 
S &= \mathcal{N}_{g}(c) \sum_{u\in S_q} \sum_{\substack{v\in S_q\\ v \ne \pm u} } \mathcal{N}_{f}(\frac{u+v}{2}) \\
&= \mathcal{N}_{g}(c) \sum_{u\in S_q} \sum_{\substack{v\in S_q\\ v \ne \pm u} } \Bk{
	p^{m-s-1} + \eta(\frac{u+v}{2})\varepsilon_f  \sqrt{p}^{m-s-1}}	\\
&= \mathcal{N}_{g}(c) \Bk{  \frac{p-1}{2}\cdot \frac{p-5}{2} p^{m-s-1} + \eta(2)\varepsilon_f  \sqrt{p}^{m-s-1}\sum_{u\in S_q} \sum_{\substack{v\in S_q\\ v \ne \pm u} } 
	\eta(u+v) }.
\end{align*}
The desired assertion then follows from Lemmas \ref{lem:N_f} and \ref{lem:eta}.	
\end{proof}
 
 \section{Main results}
 In this section, we fix $ p \equiv 1 \pmod 4 $. Let $f,\,g \in \rm{WRP}$ or $f,\,g \in \rm{WRPB}$. For $ \alpha,\,\beta \in \mathbb{F}_{q} $, we may assume from Lemma \ref{lem:Walsh-f} that $\widehat{\chi}_f(\alpha)=\varepsilon_f\sqrt{p}^{m+s}\zeta_p^{f^{\star}(\alpha)}$ and $\widehat{\chi}_g(\beta)=\varepsilon_g\sqrt{p}^{m+t}\zeta_p^{g^{\star}(\beta)}$, where $\alpha \in \mathcal{S}_f$, $\beta \in \mathcal{S}_g$, $\varepsilon_f,\varepsilon_g  \in \{\pm 1\}$ and $0 \leqslant s,t \leqslant m$. The dual functions $f^{\star}$ and $g^{\star}$ are given in Lemma \ref{lem:dual} with $l_f \in \{2,\frac{p-1}{2},p-1\}$ and $ l_g= \frac{p-1}{2} $.
  For $(a,b) \in \mathbb{F}_q^2 \backslash \{(0,0)\}$ , we define
 \begin{align} \label{N0}
 N_0 = \# \left \{(x,y) \in  \mathbb{F}_{q}^2 :
 f(x) + g(y) = 0  ,\textup{Tr}(ax + by) = 0 \right \}.
 \end{align}
 \subsection{The calculation of $ N_0 $}
To compute the weights of codewords in our codes, it suffices to determine the values of $ N_0 $ in \eqref{N0}, which are stated in Lemmas \ref{lem:N0_odd}, \ref{lem:N0_even} and \ref{lem:N0_e2}.

 \begin{lemma} \label{lem:N0_odd}
 Let $f,\,g \in \rm{WRP}$ or $f,\,g \in \rm{WRPB}$ with $ l_g= \frac{p-1}{2} $. Suppose that $s+t$ is odd and $(a,b) \in \mathbb{F}_q^2 \backslash \{(0,0)\} $. 
 	For $ (a,b)  \notin  \mathcal{S}_f  \times \mathcal{S}_g $, we have $ N_0 = p^{2m-2} $, and for $ (a,b)  \in  \mathcal{S}_f  \times \mathcal{S}_g $, we have the following assertions. 
 	When $l_f=\frac{p-1}{2}$, we have 
 		\begin{align*}
 	N_0=\begin{cases}
 	p^{2m-2},  & \textup{if }   f^{\star}(a)+g^{\star}(b)=0,\\
 	p^{2m-2}+(p-1)\varepsilon_f \varepsilon_g  \sqrt{p}^{2m+s+t-3},  & \textup{if }   f^{\star}(a)+g^{\star}(b)  \in S_q,\\
 	 p^{2m-2}-(p-1)\varepsilon_f \varepsilon_g  \sqrt{p}^{2m+s+t-3},  & \textup{if }   f^{\star}(a)+g^{\star}(b)  \in N_{sq}.
 	\end{cases}
 	\end{align*}
 	When $l_f= p-1 $, we have 	
 	\begin{align*}
 	N_0 =\begin{cases}
  p^{2m-2} +	\frac{p-1}{2} \varepsilon_f \varepsilon_g  \eta( 2   ) \sqrt{p }^{2m+s+t-3} , & \textup{if }   f^{\star}(a)\in S_q,g^{\star}(b)=\pm f^{\star}(a) ,\\
 p^{2m-2} -	\frac{p-1}{2} \varepsilon_f \varepsilon_g  \eta( 2   ) \sqrt{p }^{2m+s+t-3} , & \textup{if }   f^{\star}(a)\in N_{sq},g^{\star}(b)=\pm f^{\star}(a) ,\\
  p^{2m-2}+	(p-1)\varepsilon_f \varepsilon_g  \sqrt{p }^{2m+s+t-3}, &
  \textup{if } f^{\star}(a)+ g^{\star}(b)\in S_q,f^{\star}(a)- g^{\star}(b)\in S_q,\\
  p^{2m-2}-	(p-1)\varepsilon_f \varepsilon_g  \sqrt{p }^{2m+s+t-3}, &
  \textup{if } f^{\star}(a)+ g^{\star}(b)\in N_{sq},f^{\star}(a)- g^{\star}(b)\in N_{sq},\\  
  	p^{2m-2}, &  	\textup{otherwise}. 
 	\end{cases}
 	\end{align*}
 	When $l_f= 2 $ and $ p \equiv 1 \pmod 8 $, we have 	
 	\begin{align*}
 	N_0 
   =\begin{cases}
p^{2m-2}+	(p-1) \varepsilon_f \varepsilon_g   \sqrt{p }^{2m+s+t-3},  & \textup{if }   f^{\star}(a)= 0,g^{\star}(b)\in S_q\\
&\textup{or }   g^{\star}(b)= 0,f^{\star}(a)\in S_q,\\
p^{2m-2}-	(p-1) \varepsilon_f \varepsilon_g   \sqrt{p }^{2m+s+t-3},  & \textup{if }   f^{\star}(a)= 0,g^{\star}(b)\in N_{sq}\\
&\textup{or }   g^{\star}(b)= 0,f^{\star}(a)\in N_{sq},\\
p^{2m-2}	-2 (p-1)\varepsilon_f \varepsilon_g  \sqrt{p }^{2m+s+t-3}, & \textup{if }   f^{\star}(a)\in  S_q,g^{\star}(b)\in  S_q,\\
p^{2m-2}	+2 (p-1)\varepsilon_f \varepsilon_g  \sqrt{p }^{2m+s+t-3}, & \textup{if }   f^{\star}(a)\in  N_{sq},g^{\star}(b)\in  N_{sq},\\
p^{2m-2},&	\textup{otherwise}.
\end{cases}	
\end{align*} 
 	When $l_f= 2 $ and $ p \equiv 5 \pmod 8 $, we have 	
 	\begin{align*}
 	N_0 
 	=\begin{cases}
 p^{2m-2},  & \textup{if }   f^{\star}(a)=g^{\star}(b)=0,\\
 p^{2m-2}+	(p-1) \varepsilon_f \varepsilon_g   \sqrt{p }^{2m+s+t-3},  & \textup{if }   f^{\star}(a)= 0,g^{\star}(b)\in S_q\\
 &\textup{or }   g^{\star}(b)= 0,f^{\star}(a)\in S_q,\\
 p^{2m-2}-	(p-1) \varepsilon_f \varepsilon_g   \sqrt{p }^{2m+s+t-3},  & \textup{if }   f^{\star}(a)= 0,g^{\star}(b)\in N_{sq}\\
 &\textup{or }   g^{\star}(b)= 0,f^{\star}(a)\in N_{sq},\\
 p^{2m-2}+  \varepsilon_f \varepsilon_g \sqrt{p}^{2m+s+t-3}\eta(f^{\star}(a))\Bk{I_4\Bk{\frac{g^{\star}(b)}{f^{\star}(a)}}-\eta\Bk{\frac{g^{\star}(b)}{f^{\star}(a)}}}, &	\textup{otherwise},
 \end{cases}	
 \end{align*} 
 where $ I_4 $ is a companion sum determined in Lemma \ref{lm:companion}.
  \end{lemma} 	
 	\begin{proof}
 Let $ 2\nmid s+t $. By definition in \eqref{N0} and the orthogonal property of group characters,	 		
 		\begin{align}
 		N_0 &=\dfrac{1}{p^2}\sum_{x,y\in \mathbb{F}_{q}}\sum_{z\in \mathbb{F}_{p}}{\zeta_p}^{z(f(x)+g(y))}\sum_{h\in \mathbb{F}_{p}}{\zeta_p}^{h \textup{Tr}(ax+by)}  \nonumber \\ 
 		&=\dfrac{1}{p^2}\sum_{x,y\in \mathbb{F}_{q}}
 		\Bk{1+\sum_{z\in \mathbb{F}_{p}^*}{\zeta_p}^{z(f(x)+g(y))}}
 		\Bk{ 1+ \sum_{h\in \mathbb{F}_{p}^*}{\zeta_p}^{h \textup{Tr}(ax+by)}} \nonumber\\
 		&=p^{2m-2}+\dfrac{1}{p^2}\sum_{z\in \mathbb{F}_{p}^*}\sum_{x,y\in \mathbb{F}_{q}} {\zeta_p}^{z(f(x)+g(y)) } \nonumber \\
 		& \phantom{==}+\dfrac{1}{p^2}\sum_{x,y\in \mathbb{F}_{q}}\sum_{z\in \mathbb{F}_{p}^*}\sum_{h\in \mathbb{F}_{p}^*}{\zeta_p}^{z(f(x)+g(y))+h \textup{Tr}(ax+by)} \nonumber \\
 		&= p^{2m-2} + p^{-2}(\Delta_1+\Delta_2)  , \label{eq:N0}
 		\end{align}
 	where we write
 	\begin{align*}
 	\Delta_1 &=	\sum_{z\in \mathbb{F}_{p}^*}\sum_{x,y\in \mathbb{F}_{q}} {\zeta_p}^{z(f(x)+g(y)) }, \\
 	\Delta_2&= \sum_{x,y\in \mathbb{F}_{q}}\sum_{z\in \mathbb{F}_{p}^*}\sum_{h\in \mathbb{F}_{p}^*}{\zeta_p}^{z(f(x)+g(y))+h \textup{Tr}(ax+by)}.
 	\end{align*}
 	It follows that	
 	\begin{align*}
 	\Delta_1 &=	\sum_{z\in \mathbb{F}_{p}^*}\sigma_z\bk{\widehat{\chi}_f(0)\widehat{\chi}_g (0)}  \\
    &=\begin{cases}
    0 , & \textup{if }  f,g  \in \rm{WRPB},\\
    \varepsilon_f \varepsilon_g \sqrt{p}^{2m+s+t} \sum_{z\in \mathbb{F}_{p}^*} \eta^{s+t}(z) , & \textup{if }  f,g\in   \rm{WRP}.   
    	\end{cases}   
 	\end{align*} 
 	So we always have $ \Delta_1=0 $ when $ 2\nmid s+t $.	
Now it is sufficient to determine $ \Delta_2 $. We observe from its definition that	
 		\begin{align} 
 		\Delta_2&=\sum_{z\in \mathbb{F}_{p}^*}\sum_{h\in \mathbb{F}_{p}^*}\sum_{x\in \mathbb{F}_{q}}{\zeta_p}^{zf(x)-\textup{Tr}(hax)}\sum_{y\in \mathbb{F}_{q}}{\zeta_p}^{zg(y)-\textup{Tr}(hby)}\nonumber\\
 		&=\sum_{z\in \mathbb{F}_{p}^*}\sum_{h\in \mathbb{F}_{p}^*}\sum_{x\in \mathbb{F}_{q}}{\zeta_p}^{z(f(x)-\textup{Tr}(\frac{h}{z}ax))}\sum_{y\in \mathbb{F}_{q}}{\zeta_p}^{z(g(y)-\textup{Tr}(\frac{h}{z}by))} \nonumber \\
 		&=\sum_{z\in \mathbb{F}_{p}^*}\sum_{h\in \mathbb{F}_{p}^*} \sigma_z\Bk{\widehat{\chi}_f(ha)\widehat{\chi}_g(hb)}. \label{eq:det2}
 		\end{align}
 	Obviously, when $ (a,b)  \notin  \mathcal{S}_f  \times \mathcal{S}_g $, then from Lemma \ref{lemc} $ (ha,hb)  \notin  \mathcal{S}_f  \times \mathcal{S}_g $ for every  $ h\in \mathbb{F}_{p}^* $. Hence $ \widehat{\chi}_f(ha)=0 $ or $ \widehat{\chi}_g(hb)=0 $, and consequently by \eqref{eq:det2}
 	\[
 	\Delta_2=0.\]
 	When $ (a,  b) \in  \mathcal{S}_f \times \mathcal{S}_g  $, then $ (ha,hb)  \in  \mathcal{S}_f  \times \mathcal{S}_g $ for every  $ h\in \mathbb{F}_{p}^* $. By \eqref{eq:det2}, Lemmas \ref{lem:cyclo}, \ref{lem:Walsh-f} and \ref{lem:dual}, we obtain
	\begin{align} 
	\Delta_2 
	&=\sum_{z\in \mathbb{F}_{p}^*} \sigma_z\Bk{\sum_{h\in \mathbb{F}_{p}^*} \varepsilon_f \varepsilon_g \sqrt{p}^{2m+s+t} 
		\zeta_p^{h^{l_f}f^{\star}(a)+h^{l_g} g^{\star}(b))}		 }\nonumber\\
	&=\varepsilon_f \varepsilon_g \sqrt{p}^{2m+s+t} \sum_{z\in \mathbb{F}_{p}^*} \eta^{s+t}(z) \sigma_z\Bk{\sum_{h\in \mathbb{F}_{p}^*}  
			\zeta_p^{h^{l_f}f^{\star}(a)+h^{l_g} g^{\star}(b)}	}	.  \label{eq:det2_1}
	\end{align} 	
 	In the following, we will apply Lemmas \ref{lem:cyclo} and \ref{lem:expo sum} to determine $ \Delta_2 $ in \eqref{eq:det2_1} by considering the cases of $l_f =2 $, $l_f= \frac{p-1}{2} $ and $l_f= p-1 $, separately.
 		
 	$ (1) $	The first case we consider is $l_f= \frac{p-1}{2} $.  	
 		
 	For $ h\in \mathbb{F}_{p}^*$, clearly $h ^{\frac{p-1}{2}}=1$ if $h \in S_q$, and $ -1 $ otherwise. So we have from \eqref{eq:det2_1} that 
 	\begin{align*}
 	\Delta_2 &=\varepsilon_f \varepsilon_g \sqrt{p}^{2m+s+t} \sum_{z\in \mathbb{F}_{p}^*} \eta(z) \sigma_z\Bk{ \sum_{h\in S_q} \zeta_p^{ f^{\star}(a)+g^{\star}(b)} + \sum_{h\in N_{sq}} \zeta_p^{- f^{\star}(a)-g^{\star}(b)} }\\
 	&=\frac{p-1}{2} \varepsilon_f \varepsilon_g \sqrt{p}^{2m+s+t} 
 	\Bk{\sum_{z\in \mathbb{F}_{p}^*} \eta(z) \zeta_p^{z(f^{\star}(a)+g^{\star}(b))}
 	   + \sum_{z\in \mathbb{F}_{p}^*} \eta(-z) \zeta_p^{-z(f^{\star}(a)+g^{\star}(b))} }\\
 	&=(p-1) \varepsilon_f \varepsilon_g \sqrt{p}^{2m+s+t} 
 	 \sum_{z\in \mathbb{F}_{p}^*} \eta(z) \zeta_p^{z(f^{\star}(a)+g^{\star}(b))} \\
 	&=\begin{cases}
 	0 , & \textup{if }   f^{\star}(a)+g^{\star}(b)=0,\\
 	(p-1)\varepsilon_f \varepsilon_g \eta( f^{\star}(a)+g^{\star}(b) ) \sqrt{p}^{2m+s+t+1},  & \textup{if }   f^{\star}(a)+g^{\star}(b)  \ne  0.
 		\end{cases}
 		\end{align*}
 	 
 	$ (2) $	The second case is that $ l_f =p-1 $. 
 	
 	In this case, $ h^{p-1}=1 $ for every $ h\in \mathbb{F}_{p}^* $. By \eqref{eq:det2_1}, we have
 	\begin{align*}
 	\Delta_2 &=\varepsilon_f \varepsilon_g \sqrt{p}^{2m+s+t}
 	\sum_{z\in \mathbb{F}_{p}^*} \eta(z) \sigma_z\Bk{ \sum_{h\in S_q} \zeta_p^{ f^{\star}(a)+g^{\star}(b)} + \sum_{h\in N_{sq}} \zeta_p^{ f^{\star}(a)-g^{\star}(b)} }\\
 	&=\frac{p-1}{2} \varepsilon_f \varepsilon_g \sqrt{p}^{2m+s+t} 	\Bk{\sum_{z\in \mathbb{F}_{p}^*} \eta(z) \zeta_p^{z(f^{\star}(a)+g^{\star}(b))}
 		+ \sum_{z\in \mathbb{F}_{p}^*} \eta(z) \zeta_p^{z(f^{\star}(a)-g^{\star}(b))} }\\
  	&=\begin{cases}
 	0, & \textup{if }   f^{\star}(a)=g^{\star}(b)=0,\\
 	\frac{p-1}{2} \varepsilon_f \varepsilon_g  \eta( 2f^{\star}(a)  ) \sqrt{p}^{2m+s+t+1} , & \textup{if }   f^{\star}(a)\ne 0,g^{\star}(b)=-f^{\star}(a) ,\\
 	\frac{p-1}{2}\varepsilon_f \varepsilon_g \eta(2 f^{\star}(a)  )\sqrt{p}^{2m+s+t+1}, &\textup{if }   f^{\star}(a)\ne 0,g^{\star}(b)=f^{\star}(a) ,\\
 \frac{p-1}{2}\varepsilon_f \varepsilon_g \Bk{\eta( f^{\star}(a)+g^{\star}(b))+  \eta( f^{\star}(a)-g^{\star}(b) ) }\sqrt{p}^{2m+s+t+1}, &	\textup{otherwise}.\\
 	\end{cases}
 	\end{align*}

 	$ (3) $	The last case is that $ l_f =2 $ and it is distinguished between two subcases.
 	
 \textbf{Subcase (a)}:	If $ p \equiv 1 \pmod 8 $, then $ -1 \in C_0^{(4,p)} $. So
 	from \eqref{eq:det2_1},  
 	\begin{align*}
 	\Delta_2 &=\varepsilon_f \varepsilon_g \sqrt{p}^{2m+s+t}
 	\sum_{z\in \mathbb{F}_{p}^*} \eta(z) \sigma_z\Bk{ \sum_{h\in S_q} \zeta_p^{ h^2f^{\star}(a)+g^{\star}(b)} + \sum_{h\in N_{sq}} \zeta_p^{ h^2 f^{\star}(a)-g^{\star}(b)} }\\
 	&=\varepsilon_f \varepsilon_g \sqrt{p}^{2m+s+t}
 	\sum_{z\in \mathbb{F}_{p}^*} \eta(z) \sigma_z\Bk{ \sum_{h\in S_q} \zeta_p^{ h^2f^{\star}(a)+g^{\star}(b)} + \sum_{h\in N_{sq}} \zeta_p^{-( h^2 f^{\star}(a)+g^{\star}(b))} }\\
 &=\varepsilon_f \varepsilon_g \sqrt{p}^{2m+s+t}\Bk{
 \sum_{z\in \mathbb{F}_{p}^*} \eta(z) \sum_{h\in S_q} \zeta_p^{z( h^2f^{\star}(a)+g^{\star}(b))} +  \sum_{z\in \mathbb{F}_{p}^*} \eta(-z) \sum_{h\in N_{sq}} \zeta_p^{-z( h^2f^{\star}(a)+g^{\star}(b))}}.
   \end{align*}
   Replacing $ -z $ by $ z $ in the last double sum above, we obtain
 	\begin{align*}
 	\Delta_2 
 	&=\varepsilon_f \varepsilon_g \sqrt{p}^{2m+s+t} 
 	\sum_{z\in \mathbb{F}_{p}^*} \eta(z) \sum_{h\in \mathbb{F}_p^*} \zeta_p^{z( h^2f^{\star}(a)+g^{\star}(b))} \\
 	&=\varepsilon_f \varepsilon_g \sqrt{p}^{2m+s+t} 
 	\sum_{z\in \mathbb{F}_{p}^*} \eta(z) \zeta_p^{z g^{\star}(b)}  \sum_{h\in \mathbb{F}_p^*} \zeta_p^{z  h^2f^{\star}(a) }\\
 &=\begin{cases}
 0, & \textup{if }   f^{\star}(a)=g^{\star}(b)=0,\\
(p-1) \varepsilon_f \varepsilon_g  \eta( g^{\star}(b)  ) \sqrt{p}^{2m+s+t+1},  & \textup{if }   f^{\star}(a)= 0,g^{\star}(b)\ne 0 ,\\
(p-1) \varepsilon_f \varepsilon_g  \eta( f^{\star}(a)  ) \sqrt{p}^{2m+s+t+1} , & \textup{if }   f^{\star}(a)\ne 0,g^{\star}(b) = 0 ,\\
-(p-1)\varepsilon_f \varepsilon_g \bk{\eta( f^{\star}(a))+\eta(g^{\star}(b)) }\sqrt{p}^{2m+s+t+1}, &	\textup{otherwise}.\\
 \end{cases}	
 \end{align*}
 	
 \textbf{Subcase (b)}: 	If $ p \equiv 5 \pmod 8 $, then $ -1 \in C_2^{(4,p)} $. So
 	from \eqref{eq:det2_1},  
  	\begin{align*}
  	\Delta_2 &=\varepsilon_f \varepsilon_g \sqrt{p}^{2m+s+t}
  	\sum_{z\in \mathbb{F}_{p}^*} \eta(z) \sigma_z\Bk{ \sum_{h\in S_q} \zeta_p^{ h^2f^{\star}(a)+g^{\star}(b)} + \sum_{h\in N_{sq}} \zeta_p^{ h^2 f^{\star}(a)-g^{\star}(b)} }\\
  	&=\varepsilon_f \varepsilon_g \sqrt{p}^{2m+s+t}
  	\sum_{z\in \mathbb{F}_{p}^*} \eta(z) \sigma_z\Bk{ \sum_{h\in S_q} \zeta_p^{ h^2f^{\star}(a)+g^{\star}(b)} + \sum_{h\in S_q} \zeta_p^{-( h^2 f^{\star}(a)+g^{\star}(b))} }\\
  	&=2\varepsilon_f \varepsilon_g \sqrt{p}^{2m+s+t}
  		\sum_{z\in \mathbb{F}_{p}^*} \eta(z) \sigma_z \Bk{\sum_{h\in S_q} \zeta_p^{  h^2f^{\star}(a)+g^{\star}(b) } }\\
  	&=2\varepsilon_f \varepsilon_g \sqrt{p}^{2m+s+t}\sum_{h\in S_q}
  	\sum_{z\in \mathbb{F}_{p}^*} \eta(z)   \zeta_p^{ z( h^2f^{\star}(a)+g^{\star}(b)) }  .
  	\end{align*} 

Assume that $ f^{\star}(a)g^{\star}(b)\ne 0 $. If $\frac{g^{\star}(b)}{f^{\star}(a)} \in C_2^{(4,p)} $, then the equation $ h^2f^{\star}(a)+g^{\star}(b)=0  $ has exactly two solutions $ h_1,h_2  $ in $ S_q$, where $ h_2=-h_1 $. Otherwise if $\frac{g^{\star}(b)}{f^{\star}(a)} \notin C_2^{(4,p)} $, then the inequality $ h^2f^{\star}(a)+g^{\star}(b)\ne 0  $ holds for all $ h   $ in $ S_q$. Consequently, if $ f^{\star}(a)g^{\star}(b)\ne 0 $, then
  	\begin{align*}
  	\Delta_2   
  	&=2\varepsilon_f \varepsilon_g \sqrt{p}^{2m+s+t+1}\sum_{h\in S_q}  \eta( h^2f^{\star}(a)+g^{\star}(b))   \\
	&= \varepsilon_f \varepsilon_g \sqrt{p}^{2m+s+t+1}\sum_{h\in \mathbb{F}_p^* }  \eta( h^4 f^{\star}(a)+g^{\star}(b))   \\
	&=\varepsilon_f \varepsilon_g \sqrt{p}^{2m+s+t+1}\eta(f^{\star}(a))\Bk{I_4\Bk{\frac{g^{\star}(b)}{f^{\star}(a)}}-\eta\Bk{\frac{g^{\star}(b)}{f^{\star}(a)}}},
  	\end{align*}
  	where $ I_4 $ is determined from Lemma \ref{lm:companion}.	
Thus we conclude that
    \begin{align*}
    \Delta_2 
    =\begin{cases}
    0, & \textup{if }   f^{\star}(a)=g^{\star}(b)=0,\\
    (p-1) \varepsilon_f \varepsilon_g  \eta( g^{\star}(b)  ) \sqrt{p}^{2m+s+t+1},  & \textup{if }   f^{\star}(a)= 0,g^{\star}(b)\ne 0 ,\\
    (p-1) \varepsilon_f \varepsilon_g  \eta( f^{\star}(a)  ) \sqrt{p}^{2m+s+t+1},  & \textup{if }   f^{\star}(a)\ne 0,g^{\star}(b) = 0 ,\\
   \varepsilon_f \varepsilon_g \sqrt{p}^{2m+s+t+1}\eta(f^{\star}(a))\Bk{I_4\Bk{\frac{g^{\star}(b)}{f^{\star}(a)}}-\eta\Bk{\frac{g^{\star}(b)}{f^{\star}(a)}}}, &	\textup{otherwise}.
    \end{cases}	
    \end{align*}

 	The desired conclusion then follows from \eqref{eq:N0}, completing the proof.
 	\end{proof}

 \begin{lemma} \label{lem:N0_even}
 Let $f,\,g \in \rm{WRP}$ with $ l_g= \frac{p-1}{2} $. Suppose that $s+t$ is even and $(a,b) \in \mathbb{F}_q^2 \backslash \{(0,0)\} $. 
 	For $ (a,b)  \notin  \mathcal{S}_f  \times \mathcal{S}_g $,	we have $ N_0 = p^{2m-2} +(p-1)\varepsilon_f \varepsilon_g  \sqrt{p}^{2m+s+t-4} $, and for $ (a,b)  \in  \mathcal{S}_f  \times \mathcal{S}_g $, we have the following assertions. 
 	When $l_f=\frac{p-1}{2} $, we have 
 	\begin{align*}
 	N_0=\begin{cases}
 p^{2m-2} +(p-1)\varepsilon_f \varepsilon_g  \sqrt{p}^{2m+s+t-2} , & \textup{if }   f^{\star}(a)+g^{\star}(b)=0,\\
 p^{2m-2},  & \textup{if }   f^{\star}(a)+g^{\star}(b)  \ne  0.
 	\end{cases}
 	\end{align*}
 	When $l_f= p-1 $, we have 	
 	\begin{align*}
 	N_0=\begin{cases}
  p^{2m-2} +(p-1)\varepsilon_f \varepsilon_g  \sqrt{p}^{2m+s+t-2} , & \textup{if }   f^{\star}(a)=g^{\star}(b)=0,\\
 p^{2m-2} +	\frac{p-1}{2}  \varepsilon_f \varepsilon_g    \sqrt{p}^{2m+s+t-2} , & \textup{if }   f^{\star}(a)\ne 0,g^{\star}(b)=-f^{\star}(a) \\
 	&\textup{or }   f^{\star}(a)\ne 0,g^{\star}(b)=f^{\star}(a) ,\\
 p^{2m-2}, &	\textup{otherwise}.
 	\end{cases}
 	\end{align*}
 	When $l_f= 2 $ and $ p \equiv 1 \pmod 8 $, we have 	
 	\begin{align*}
 	N_0 
 =\begin{cases}
  p^{2m-2} +(p-1)\varepsilon_f \varepsilon_g  \sqrt{p}^{2m+s+t-2}, & \textup{if }   f^{\star}(a)=g^{\star}(b)=0,\\
  p^{2m-2} +2 \varepsilon_f \varepsilon_g  \sqrt{p}^{2m+s+t-2}, &	\textup{if } f^{\star}(a) g^{\star}(b) \in S_q,\\
  p^{2m-2} ,  & \textup{otherwise}.
 \end{cases}
 	\end{align*} 
 	When $l_f= 2 $ and $ p \equiv 5 \pmod 8 $, we have 	
 	\begin{align*}
 	N_0 
 	=\begin{cases}
 	 p^{2m-2} +(p-1)\varepsilon_f \varepsilon_g  \sqrt{p}^{2m+s+t-2}, & \textup{if }   f^{\star}(a)=g^{\star}(b)=0,\\
 	 p^{2m-2} + 4 \varepsilon_f \varepsilon_g \sqrt{p}^{2m+s+t-2 },  & \textup{if }   \frac{g^{\star}(b)}{f^{\star}(a)} \in C_2^{(4,p)} ,\\
 	 p^{2m-2}  ,   &	\textup{otherwise}.
 	\end{cases}
 	\end{align*} 
 \end{lemma} 	
 \begin{proof}
The proof is completed in a manner analogous to the previous lemma by noting that $ 2\mid s+t $. From \eqref{eq:N0} and \eqref{eq:det2},	 		
 	\begin{align*}
 	N_0 = p^{2m-2} + p^{-2}(\Delta_1+\Delta_2)  ,  
 	\end{align*}
 	where  
 	\begin{align*}
 	\Delta_1 & =(p-1)\varepsilon_f \varepsilon_g \sqrt{p}^{2m+s+t},\\   
 	\Delta_2&= \sum_{z\in \mathbb{F}_{p}^*}\sum_{h\in \mathbb{F}_{p}^*} \sigma_z\Bk{\widehat{\chi}_f(ha)\widehat{\chi}_g(hb)}. 
 	\end{align*}
 We always have $ 	\Delta_2=0 $ unless $ (a,  b) \in  \mathcal{S}_f \times \mathcal{S}_g  $. Now we let $ (a,  b) \in  \mathcal{S}_f \times \mathcal{S}_g  $. Then the value of $ \Delta_2 $ in \eqref{eq:det2_1} is determined by distinguishing the cases of $l_f =2 $, $l_f= \frac{p-1}{2} $ and $l_f= p-1 $, respectively.
 	
 	$ (1) $	The first case we consider is $l_f= \frac{p-1}{2} $.  	
 	
 	It follows from \eqref{eq:det2_1} that 
 	\begin{align*}
 	\Delta_2
 	&=(p-1) \varepsilon_f \varepsilon_g \sqrt{p}^{2m+s+t} 
 	\sum_{z\in \mathbb{F}_{p}^*}  \zeta_p^{z(f^{\star}(a)+g^{\star}(b))} \\
 	&=\begin{cases}
 	(p-1)^2 \varepsilon_f \varepsilon_g \sqrt{p}^{2m+s+t}  , & \textup{if }   f^{\star}(a)+g^{\star}(b)=0,\\
 	- (p-1)\varepsilon_f \varepsilon_g \sqrt{p}^{2m+s+t},  & \textup{if }   f^{\star}(a)+g^{\star}(b)  \ne  0.
 	\end{cases}
 	\end{align*}
 	
 	$ (2) $	The second case is that $ l_f =p-1 $. 
 	
 	Again from \eqref{eq:det2_1}, we have
 	\begin{align*}
 	\Delta_2 
 	&=\frac{p-1}{2} \varepsilon_f \varepsilon_g \sqrt{p}^{2m+s+t} 	\Bk{\sum_{z\in \mathbb{F}_{p}^*}  \zeta_p^{z(f^{\star}(a)+g^{\star}(b))}
 		+ \sum_{z\in \mathbb{F}_{p}^*}   \zeta_p^{z(f^{\star}(a)-g^{\star}(b))} }\\
 	&=\begin{cases}
 	(p-1)^2 \varepsilon_f \varepsilon_g \sqrt{p}^{2m+s+t}, & \textup{if }   f^{\star}(a)=g^{\star}(b)=0,\\
 	\frac{p-1}{2}(p-2) \varepsilon_f \varepsilon_g    \sqrt{p}^{2m+s+t} , & \textup{if }   f^{\star}(a)\ne 0,g^{\star}(b)=-f^{\star}(a) \\
   &\textup{or }   f^{\star}(a)\ne 0,g^{\star}(b)=f^{\star}(a) ,\\
 -(p-1)\varepsilon_f \varepsilon_g \sqrt{p}^{2m+s+t}, &	\textup{otherwise}.
 	\end{cases}
 	\end{align*}
 	
 	$ (3) $	The last case is that $ l_f =2 $ and we need only consider two different subcases.
 	
 	\textbf{Subcase (a)}:	If $ p \equiv 1 \pmod 8 $, then  
 	from \eqref{eq:det2_1},  
 	\begin{align*}
 	\Delta_2 
 	&=\varepsilon_f \varepsilon_g \sqrt{p}^{2m+s+t} 
 	\sum_{z\in \mathbb{F}_{p}^*}   \sum_{h\in \mathbb{F}_p^*} \zeta_p^{z( h^2f^{\star}(a)+g^{\star}(b))} \\
 	&=\varepsilon_f \varepsilon_g \sqrt{p}^{2m+s+t} 
 	\sum_{z\in \mathbb{F}_{p}^*}  \zeta_p^{z g^{\star}(b)}  \sum_{h\in \mathbb{F}_p^*} \zeta_p^{z  h^2f^{\star}(a) }\\
 	&=\begin{cases}
 	(p-1)^2 \varepsilon_f \varepsilon_g \sqrt{p}^{2m+s+t}, & \textup{if }   f^{\star}(a)=g^{\star}(b)=0,\\
 	 (p+1)\varepsilon_f \varepsilon_g  \sqrt{p}^{2m+s+t }, &	\textup{if } f^{\star}(a) g^{\star}(b) \in S_q,\\
 	 - (p-1) \varepsilon_f \varepsilon_g    \sqrt{p}^{2m+s+t},  & \textup{otherwise}.
 	 \end{cases}	
 	\end{align*}
 	
 	\textbf{Subcase (b)}: 	If $ p \equiv 5 \pmod 8 $, then 
 	from \eqref{eq:det2_1},  
 	\begin{align*}
 	\Delta_2 
 	=2\varepsilon_f \varepsilon_g \sqrt{p}^{2m+s+t}\sum_{h\in S_q}
 	\sum_{z\in \mathbb{F}_{p}^*}   \zeta_p^{ z( h^2f^{\star}(a)+g^{\star}(b)) }  .
 	\end{align*} 
 The value of $ \Delta_2  $	is clear if $ f^{\star}(a)g^{\star}(b)= 0 $. We now assume that $ f^{\star}(a)g^{\star}(b)\ne 0 $. If $\frac{g^{\star}(b)}{f^{\star}(a)} \in C_2^{(4,p)} $, then the equation $ h^2f^{\star}(a)+g^{\star}(b)=0  $ has exactly two solutions $ h_1,h_2  $ in $ S_q$, where $ h_2=-h_1 $. Hence,
 	\begin{align*}
 	\Delta_2   
 	&=2\varepsilon_f \varepsilon_g \sqrt{p}^{2m+s+t } \Bk{2(p-1)-(\frac{p-1}{2}-2)}  \\
 	&=(3p+1)\varepsilon_f \varepsilon_g \sqrt{p}^{2m+s+t } .
 	\end{align*}	
 	 Otherwise if $\frac{g^{\star}(b)}{f^{\star}(a)} \notin C_2^{(4,p)} $, then the inequality $ h^2f^{\star}(a)+g^{\star}(b)\ne 0  $ holds for all $ h $ in $ S_q$.
 	 Thus  
 	\begin{align*}
 	 \Delta_2   
 	 &=2\varepsilon_f \varepsilon_g \sqrt{p}^{2m+s+t } \times \frac{p-1}{2}\times (-1)  \\
 	 &=-(p-1)\varepsilon_f \varepsilon_g \sqrt{p}^{2m+s+t } .
 	 \end{align*}
 	 So we conclude that
 	\begin{align*}
 	\Delta_2 
 	=\begin{cases}
 	(p-1)^2 \varepsilon_f \varepsilon_g \sqrt{p}^{2m+s+t}, & \textup{if }   f^{\star}(a)=g^{\star}(b)=0,\\
 	(3p+1)\varepsilon_f \varepsilon_g \sqrt{p}^{2m+s+t },  & \textup{if }   \frac{g^{\star}(b)}{f^{\star}(a)} \in C_2^{(4,p)} ,\\
 	-(p-1) \varepsilon_f \varepsilon_g   \sqrt{p}^{2m+s+t},   &	\textup{otherwise}.
 	\end{cases}	
 	\end{align*} 	
 	
 	The desired conclusion then follows from \eqref{eq:N0}, completing the proof.
 \end{proof}

 \begin{lemma} \label{lem:N0_e2}
 	Let $f,\,g \in \rm{WRPB}$ with $ l_g= \frac{p-1}{2} $. Suppose that $s+t$ is even and $(a,b) \in \mathbb{F}_q^2 \backslash \{(0,0)\} $. 
 	For $ (a,b)  \notin  \mathcal{S}_f  \times \mathcal{S}_g $, we have $ N_0 = p^{2m-2} $, and for $ (a,b)  \in  \mathcal{S}_f  \times \mathcal{S}_g $, we have the following assertions. 
 	When $l_f=\frac{p-1}{2} $, we have 
 	\begin{align*}
 	N_0=\begin{cases}
 p^{2m-2}+	(p-1)^2 \varepsilon_f \varepsilon_g \sqrt{p}^{2m+s+t-4}  , & \textup{if }   f^{\star}(a)+g^{\star}(b)=0,\\
 p^{2m-2}	- (p-1)\varepsilon_f \varepsilon_g \sqrt{p}^{2m+s+t-4},  & \textup{if }   f^{\star}(a)+g^{\star}(b)  \ne  0.
 	\end{cases}
 	\end{align*}
 	When $l_f= p-1 $, we have 	
 	\begin{align*}
 	N_0=\begin{cases}
 p^{2m-2}+(p-1)^2 \varepsilon_f \varepsilon_g \sqrt{p}^{2m+s+t-4}, & \textup{if }   f^{\star}(a)=g^{\star}(b)=0,\\
 p^{2m-2}+\frac{p-1}{2}(p-2) \varepsilon_f \varepsilon_g    \sqrt{p}^{2m+s+t-4} , & \textup{if }   f^{\star}(a)\ne 0,g^{\star}(b)=-f^{\star}(a) \\
&\textup{or }   f^{\star}(a)\ne 0,g^{\star}(b)=f^{\star}(a) ,\\
 p^{2m-2}-(p-1)\varepsilon_f \varepsilon_g \sqrt{p}^{2m+s+t-4}, &	\textup{otherwise}.
\end{cases} 
 	\end{align*}
 	When $l_f= 2 $ and $ p \equiv 1 \pmod 8 $, we have 	
 	\begin{align*}
 	N_0 
 	=\begin{cases}
 	p^{2m-2} +	(p-1)^2 \varepsilon_f \varepsilon_g  \sqrt{p}^{2m+s+t-4}, & \textup{if }   f^{\star}(a)=g^{\star}(b)=0,\\
 	p^{2m-2}+	(p+1)\varepsilon_f \varepsilon_g  \sqrt{p}^{2m+s+t-4}, &	\textup{if } f^{\star}(a) g^{\star}(b) \in S_q,\\
 	p^{2m-2}- (p-1) \varepsilon_f \varepsilon_g   \sqrt{p}^{2m+s+t-4},  & \textup{otherwise}.
 	\end{cases}
 	\end{align*} 
 	When $l_f= 2 $ and $ p \equiv 5 \pmod 8 $, we have 	
 	\begin{align*}
 	N_0 
 	=\begin{cases}
 	p^{2m-2} +(p-1)^2 \varepsilon_f \varepsilon_g  \sqrt{p}^{2m+s+t-4}, & \textup{if }   f^{\star}(a)=g^{\star}(b)=0,\\
 p^{2m-2} +	(3p+1)\varepsilon_f \varepsilon_g  \sqrt{p}^{2m+s+t-4},  & \textup{if }   \frac{g^{\star}(b)}{f^{\star}(a)} \in C_2^{(4,p)} ,\\
 p^{2m-2} 	-(p-1) \varepsilon_f \varepsilon_g  \sqrt{p}^{2m+s+t-4},   &	\textup{otherwise}.
 	\end{cases}
 	\end{align*} 
 \end{lemma} 	 
 \begin{proof}
 	Notation that $ \Delta_1 =	\sum_{z\in \mathbb{F}_{p}^*}\sigma_z\bk{\widehat{\chi}_f(0)\widehat{\chi}_g (0)} =0 $ for $f,\,g \in \rm{WRPB}$. From \eqref{eq:N0}, $ N_0= p^{2m-2} + p^{-2}(\Delta_1+\Delta_2) = p^{2m-2} + p^{-2} \Delta_2   $, where $ \Delta_2  $ is determined in Lemma \ref{lem:N0_even}. This completes the proof. 
 \end{proof}

 \subsection{Weight distributions of $C_{D_{f,g}}$ from two weakly regular plateaued functions}   
 
In this subsection, we will deal with the weight distributions of $ C_{D_{f,g}} $ defined by \eqref{eq:code} and \eqref{def:D-fg}.
Its length is denoted by $ n $ and is given in Lemma \ref{lem:length}.
We will show the main results in the following theorems explicitly. For abbreviation, we write $\tau=2m+s+t$ and $\gamma =2m-s-t$.
 
 \begin{theorem}\label{th:s+todd}
Suppose that $f,\,g \in \rm{WRP}$ or $f,\,g \in \rm{WRPB}$ with $ l_g= \frac{p-1}{2} $. Let $s+t$ be odd. 
 	If $l_f=\frac{p-1}{2}$, then $C_{D_{f,g}}$ is a three-weight $[p^{2m-1}-1,2m]$ linear code with its weight distribution listed in Table \ref{odd1}. 
 	If $l_f= p-1 $, then $C_{D_{f,g}}$ is a five-weight $[p^{2m-1}-1,2m]$ linear code with its weight distribution listed in Table \ref{odd2}. 
  If $ l_f =2 $ and $ p\equiv 1 \pmod 8 $, then $C_{D_{f,g}}$ is a five-weight $[p^{2m-1}-1,2m]$ linear code with its weight distribution listed in Table \ref{odd3}. If $ l_f =2 $ and $ p\equiv 5 \pmod 8 $, then $C_{D_{f,g}}$ is a $[p^{2m-1}-1,2m]$ linear code with its weight distribution listed in Table \ref{odd4}. 
  \end{theorem}
 	\begin{proof}
 		From Lemma \ref{lem:length}, the length is $n= p^{2m-1}-1$. Let $(a,b) \in \mathbb{F}_q^2 \backslash \{(0,0)\}$, the weight of nonzero codewords $ \textbf{c}(a,b) $ is denoted by $\texttt{wt}(\textbf{c}(a,b))$. It follows that \[\texttt{wt}(\textbf{c}(a,b))=n+1 - N_0  ,
 		\] where $N_0$ is given by Lemma \ref{lem:N0_odd}.
 	To be more precise, for each $  (a,  b) \notin  \mathcal{S}_f \times \mathcal{S}_g $, we have 
 	\begin{align*}	
 	\texttt{wt}(\textbf{c}(a,b))=(p-1)p^{2m-2}.  
 	\end{align*}
     
 For each $  (a,  b) \in  \mathcal{S}_f \times \mathcal{S}_g $, there are four different cases when $ \texttt{wt}(\textbf{c}(a,b))\ne (p-1)p^{2m-2} $.		
 	
 	$ (1) $	When $l_f=\frac{p-1}{2}$, we have 
 		\begin{align*}
 	\texttt{wt}(\textbf{c}(a,b))=\begin{cases} 
 		(p-1)\bk{p^{2m-2}-  \varepsilon_f \varepsilon_g  \sqrt{p }^{\tau-3} }  ,  & \frac{p-1}{2}\mathcal{T}(i) \textup{ times},\\
 		(p-1)\bk{p^{2m-2}+  \varepsilon_f \varepsilon_g  \sqrt{p }^{\tau-3} },  & \frac{p-1}{2}\mathcal{T}(j) \textup{ times},
 		\end{cases}
 		\end{align*}
 		where $ \mathcal{T}(i) $ and $ \mathcal{T}(j) $ are computed in Lemma \ref{lem:T} for $ i \in S_q$ and $ j \in N_{sq}$. This leads to the weight distribution in Table \ref{odd1}.

 	$ (2) $	When $l_f= p-1 $, we have 	
 		\begin{align*}
 	\texttt{wt}(\textbf{c}(a,b)) =\begin{cases}
 	(p-1)\bk{p^{2m-2}-\frac{1}{2} \varepsilon_f \varepsilon_g  \eta( 2   ) \sqrt{p }^{\tau-3}},  & E_1 \textup{ times} ,\\
 	(p-1)\bk{p^{2m-2}+\frac{1}{2} \varepsilon_f \varepsilon_g  \eta( 2   ) \sqrt{p }^{\tau-3}}, & E_2 \textup{ times} ,\\
 	(p-1)\bk{p^{2m-2}-  \varepsilon_f \varepsilon_g    \sqrt{p }^{\tau-3}}, &
 		B_{S_q} \textup{ times},\\
 	(p-1)\bk{p^{2m-2}+ \varepsilon_f \varepsilon_g   \sqrt{p }^{\tau-3}} , &
 		B_{N_{sq}} \textup{ times}, 
 		\end{cases}
 		\end{align*}
 	where the numbers $ B_{S_q}  $ and $  B_{N_{sq}}  $ are computed in Lemma \ref{lem:S_q-N_{sq}}, and
 \begin{align*}
 	E_1 & = \#\{ (a,  b) \in  \mathcal{S}_f \times \mathcal{S}_g : f^{\star}(a)\in S_q,g^{\star}(b)=\pm f^{\star}(a) \} 	=(p-1)\mathcal{N}_{f}(i) \mathcal{N}_{g}(i),\\
 	E_2 & = \#\{ (a,  b) \in  \mathcal{S}_f \times \mathcal{S}_g :f^{\star}(a)\in N_{sq},g^{\star}(b)=\pm f^{\star}(a) \} 	=(p-1)\mathcal{N}_{f}(j) \mathcal{N}_{g}(j),
 \end{align*}
 for $ i \in S_q$, $ j \in N_{sq}$, and $ \mathcal{N}_{f} $ and $ \mathcal{N}_{g} $ are given in Lemma \ref{lem:N_f}. The weight distribution in Table \ref{odd2} is then established.

 $ (3) $ When $l_f= 2 $ and $ p \equiv 1 \pmod 8 $, we have 	
 \begin{align*}
 	\texttt{wt}(\textbf{c}(a,b))
 	=\begin{cases}
 	(p-1)\bk{p^{2m-2}-  \varepsilon_f \varepsilon_g \sqrt{p}^{\tau-3} } ,  & E_3 \textup{ times},\\
 	(p-1)\bk{p^{2m-2}+  \varepsilon_f \varepsilon_g \sqrt{p}^{\tau-3} } ,  & E_4 \textup{ times},\\
 	(p-1)\bk{p^{2m-2}+2  \varepsilon_f \varepsilon_g \sqrt{p}^{\tau-3} } , & \frac{(p-1)^2}{4}  \mathcal{N}_{f}(i)\mathcal{N}_{g}(i) \textup{ times} , \\
 	(p-1)\bk{p^{2m-2}-2  \varepsilon_f \varepsilon_g \sqrt{p}^{\tau-3} } , & \frac{(p-1)^2}{4}  \mathcal{N}_{f}(j)\mathcal{N}_{g}(j) \textup{ times}, 
 		\end{cases}	
 		\end{align*} 
 where
 	\begin{align*}		
 E_3 & = \#\{ (a,  b) \in  \mathcal{S}_f \times \mathcal{S}_g : f^{\star}(a)= 0,g^{\star}(b)\in S_q\textup{ or } g^{\star}(b)= 0,f^{\star}(a)\in S_q\}\\
 &=\frac{p-1}{2} \bk{\mathcal{N}_{f}(0)\mathcal{N}_{g}(i)+\mathcal{N}_{f}(i)\mathcal{N}_{g}(0)}, \\
 E_4 & = \#\{ (a,  b) \in  \mathcal{S}_f \times \mathcal{S}_g : f^{\star}(a)= 0,g^{\star}(b)\in N_{sq}\textup{ or }  g^{\star}(b)= 0, f^{\star}(a)\in N_{sq}\}\\
 &=\frac{p-1}{2} \bk{\mathcal{N}_{f}(0)\mathcal{N}_{g}(j)+\mathcal{N}_{f}(j)\mathcal{N}_{g}(0)},
 	\end{align*} 
 for $ i \in S_q$ and $ j \in N_{sq}$. Thus we get the weight distribution listed in Table \ref{odd3}.

 	$ (4) $	When $l_f= 2 $ and $ p \equiv 5 \pmod 8 $, we have 	
 		\begin{align*}
 	\texttt{wt}(\textbf{c}(a,b))
 		=\begin{cases}
 	(p-1)\bk{p^{2m-2}-  \varepsilon_f \varepsilon_g \sqrt{p}^{\tau-3} } ,  & E_3 \textup{ times},\\
 	(p-1)\bk{p^{2m-2}+  \varepsilon_f \varepsilon_g \sqrt{p}^{\tau-3} },  & E_4 \textup{ times},\\
 		\makecell{ (p-1)p^{2m-2}- \varepsilon_f \varepsilon_g \sqrt{p}^{\tau-3} \eta(u)\Bk{I_4\bk{\frac{v}{u}}-\eta\bk{\frac{v}{u}}}     \\
 		\textup{ \qquad \qquad	for all }  u,v \in \mathbb{F}_p^* , } & \mathcal{N}_{f}(u)\mathcal{N}_{g}(v) \textup{ times} .
 		\end{cases}	
 		\end{align*} 
 	The weight distribution of this case is summarized in Table \ref{odd4}. 		
 	\end{proof}
 	
 	\begin{table}[htbp]
 		\centering
 		\caption[short text]{The weight distribution of $C_{D_{f,g}}$ in Theorem \ref{th:s+todd} when $ l_f=\frac{p-1}{2}$.}
 		\label{odd1}
 		\begin{tabular}{ll}
 			\hline
 			weight & frequency \\ \hline
 			$ 0 $ & $1 $ \\
 			$ (p-1)p^{2m-2}$  & $p^{2m}-(p-1)p^{\gamma-1}-1 \qquad $ \\
 			$ (p-1)\bk{p^{2m-2}-     \sqrt{p }^{\tau-3} } $ & $\frac{p-1}{2}\bk{p^{\gamma-1}+  \sqrt{p }^{\gamma-1}} $\\
 			$ (p-1)\bk{p^{2m-2}+     \sqrt{p }^{\tau-3} } $  & $ \frac{p-1}{2}\bk{p^{\gamma-1}-  \sqrt{p }^{\gamma-1}} $ \\
 			\hline
 		\end{tabular}
 	\end{table}

 	\begin{table}[htbp]
 	\centering
 	\caption[short text] {The weight distribution of $C_{D_{f,g}}$ in Theorem \ref{th:s+todd} when $ l_f =p-1$.}
 		\label{odd2}
 		\begin{tabular}{ll}
 			\hline
 			weight & frequency \\ \hline
 				$ 0 $ & $1 $ \\
 			$ (p-1)p^{2m-2}$  & $p^{2m}-1-E_1-E_2-B_{S_q}-B_{N_{sq}} $ \\
 			$ (p-1)\bk{p^{2m-2}-\frac{1}{2} \varepsilon_f \varepsilon_g  \eta( 2   ) \sqrt{p }^{\tau-3}} $
 			& $ E_1$ \\
 			$(p-1)\bk{p^{2m-2}+\frac{1}{2} \varepsilon_f \varepsilon_g  \eta( 2   ) \sqrt{p }^{\tau-3}}$
 			&$ E_2 $\\
 			$(p-1)\bk{p^{2m-2}-  \varepsilon_f \varepsilon_g    \sqrt{p }^{\tau-3}}$  & $ B_{S_q} $ \\
 		$ 	(p-1)\bk{p^{2m-2}+ \varepsilon_f \varepsilon_g   \sqrt{p }^{\tau-3}} $ & $ B_{N_{sq}} $ \\
 			\hline
 		\end{tabular}
  	\end{table}

 	\begin{table}[htbp]
 	\centering
 	\caption[short text]{The weight distribution of $C_{D_{f,g}}$ in Theorem \ref{th:s+todd} when $ l_f =2 $ and $ p\equiv 1 \pmod 8 $.}
 		\label{odd3}
 		\begin{tabular}{ll}
 			\hline
 			weight & frequency \\ \hline
 				$ 0 $ & $1 $ \\
 			$ (p-1)p^{2m-2}$  & $p^{2m} -1- E_3 -E_4 -\frac{(p-1)^2}{4} \bk{ \mathcal{N}_{f}(i)\mathcal{N}_{g}(i)+\mathcal{N}_{f}(j)\mathcal{N}_{g}(j)}$ \\
 			$ (p-1)\bk{p^{2m-2}-  \varepsilon_f \varepsilon_g \sqrt{p}^{\tau-3} }  $
 			& $ E_3$\\
 			$ (p-1)\bk{p^{2m-2}+  \varepsilon_f \varepsilon_g \sqrt{p}^{\tau-3} }  $
 			& $E_4$\\
 			$ (p-1)\bk{p^{2m-2}+2  \varepsilon_f \varepsilon_g \sqrt{p}^{\tau-3} }  $  & $ \frac{(p-1)^2}{4}  \mathcal{N}_{f}(i)\mathcal{N}_{g}(i) $\\
  			$ (p-1)\bk{p^{2m-2}-2  \varepsilon_f \varepsilon_g \sqrt{p}^{\tau-3} }  $  & $ \frac{(p-1)^2}{4}  \mathcal{N}_{f}(j)\mathcal{N}_{g}(j) $ \\
  			\hline
 		\end{tabular}

	\end{table}

 \begin{table}[htbp]
 	\centering
 	\caption[short text]{The weight distribution of $C_{D_{f,g}}$ in Theorem \ref{th:s+todd} when $ l_f =2 $ and $ p\equiv 5 \pmod 8 $.}
 	\label{odd4}
 	\begin{tabular}{ll}
 		\hline
 		weight & frequency \\ \hline
 		$ 0 $ & $1 $ \\
 		$ (p-1)p^{2m-2}$  & $p^{2m} -1-E_3-E_4 -\sum_{u,v \in \mathbb{F}_p^*} \mathcal{N}_{f}(u)\mathcal{N}_{g}(v) $ \\
 		$ (p-1)\bk{p^{2m-2}-  \varepsilon_f \varepsilon_g \sqrt{p}^{\tau-3} }  $
 		& $ E_3 $\\
 		$ (p-1)\bk{p^{2m-2}+  \varepsilon_f \varepsilon_g \sqrt{p}^{\tau-3} }  $
 		& $ E_4 $\\
 		\makecell{$ (p-1)p^{2m-2}- \varepsilon_f \varepsilon_g \sqrt{p}^{\tau-3} \eta(u)\Bk{I_4\bk{\frac{v}{u}}-\eta\bk{\frac{v}{u}}}    $ \\
 			for all $ u,v \in \mathbb{F}_p^* $ } & $    \mathcal{N}_{f}(u)\mathcal{N}_{g}(v) $ \\ 
 		\hline
 	\end{tabular}
 \end{table}
 \begin{theorem}\label{th:s+teven1}
  Suppose that $f,\,g \in \rm{WRP}$ with $ l_g= \frac{p-1}{2} $. Let $s+t$ be even. 
 If $l_f=\frac{p-1}{2} $, then $C_{D_{f,g}}$ is a three-weight $[n,2m]$ linear code with its weight distribution listed in Table \ref{evenc1}. If $l_f= p-1 $, then $C_{D_{f,g}}$ is a four-weight $[n,2m]$ linear code with its weight distribution listed in Table \ref{evenc2}. Otherwise if $ l_f =2 $, then $C_{D_{f,g}}$ is a four-weight $[n,2m]$ linear code with its weight distribution listed in Table \ref{evenc3} when $ p \equiv 1 \pmod 8 $, and in Table \ref{evenc4} when $ p \equiv 5 \pmod 8 $. Here we set $ n= p^{2m-1}-1+(p-1) \varepsilon_f \varepsilon_g \sqrt{p}^{\tau-2} $ for brevity.
\end{theorem} 
 
 \begin{proof}
 The length of the code $C_{D_{f,g}}$ comes from Lemma \ref{lem:length}. 
For $(a,b) \in \mathbb{F}_q^2 \backslash \{(0,0)\}$, the weight $ \texttt{wt}(\textbf{c}(a,b))=n+1 - N_0  $ can be obtained from Lemma \ref{lem:N0_even}.
 To be more explicit, when $  (a,  b) \notin  \mathcal{S}_f \times \mathcal{S}_g $, we have 
 \begin{align*}	
 \texttt{wt}(\textbf{c}(a,b))=(p-1)\bk{p^{2m-2}+ (p-1)\varepsilon_f \varepsilon_g  \sqrt{p}^{\tau-4} }.  
 \end{align*}
By Lemma \ref{lem:car}, the frequency of such codewords equals $ p^{2m}-p^{\gamma} $ since $ f,\,g \in \rm{WRP} $.
When $  (a,  b) \in  \mathcal{S}_f \times \mathcal{S}_g \backslash \{(0,0)\} $, we will discuss the following four different cases.		
 
 $ (1) $	When $l_f=\frac{p-1}{2}$, we have 
 \begin{align*}
 \texttt{wt}(\textbf{c}(a,b))=\begin{cases} 
  (p-1)p^{2m-2},  & \mathcal{T}(0)-1 \textup{ times} , \\
  (p-1)\bk{p^{2m-2}+\varepsilon_f \varepsilon_g \sqrt{p}^{\tau-2}},   &  (p-1)\mathcal{T}(c) \textup{ times} , 
 \end{cases}
 \end{align*}
 where $ \mathcal{T}(0) $ and $ \mathcal{T}(c) $ are given in Lemma \ref{lem:T} for $ c \ne 0 $. This gives the weight distribution in Table \ref{evenc1}. 		
 
 $ (2) $ When $l_f= p-1 $, we have 	
 \begin{align*}
 \texttt{wt}(\textbf{c}(a,b)) =\begin{cases}
 	(p-1)p^{2m-2},  & \mathcal{N}_{f}(0)\mathcal{N}_{g}(0)-1  \textup{ times}, \\
 	 (p-1)\bk{p^{2m-2}+\frac{1}{2}\varepsilon_f \varepsilon_g \sqrt{p}^{\tau-2}} , &  F_1 \textup{ times},\\
  (p-1)\bk{p^{2m-2}+ \varepsilon_f \varepsilon_g \sqrt{p}^{\tau-2}} , &  F_2 \textup{ times}, 
 \end{cases}
 \end{align*}
 where we define  
 \begin{align*}
 F_1 & = \#\{ (a,  b) \in  \mathcal{S}_f \times \mathcal{S}_g :f^{\star}(a)\ne 0,g^{\star}(b)=\pm f^{\star}(a) \} 	=2\sum_{c\in \mathbb{F}_p^*} \mathcal{N}_{f}(c)\mathcal{N}_{g}(c)   ,\\
 F_2 & = p^{\gamma }-\mathcal{N}_{f}(0)\mathcal{N}_{g}(0)-F_1.
 \end{align*} 
 Thus we obtain the weight distribution in Table \ref{evenc2}.

 $ (3) $ When $l_f= 2 $ and $ p \equiv 1 \pmod 8 $, we have 	
 \begin{align*}
 \texttt{wt}(\textbf{c}(a,b))
 =\begin{cases}
(p-1)p^{2m-2},  & \mathcal{N}_{f}(0)\mathcal{N}_{g}(0)-1  \textup{ times}, \\
(p-1) p^{2m-2}+(p-3)\varepsilon_f \varepsilon_g \sqrt{p}^{\tau-2} , &  F_3 \textup{ times},  \\
 (p-1)\bk{p^{2m-2}+ \varepsilon_f \varepsilon_g \sqrt{p}^{\tau-2}} , &  F_4 \textup{ times},
 \end{cases}	
 \end{align*} 
 where
 \begin{align*}		
 F_3 & = \#\{ (a,  b) \in  \mathcal{S}_f \times \mathcal{S}_g : f^{\star}(a) g^{\star}(b) \in S_q\}
 =\frac{(p-1)^2}{4}(\mathcal{N}_{f}(i)\mathcal{N}_{g}(i)+\mathcal{N}_{f}(j)\mathcal{N}_{g}(j)),   \\
 F_4 & = p^{\gamma }-\mathcal{N}_{f}(0)\mathcal{N}_{g}(0)-F_3,
 \end{align*} 
 for $ i \in S_q$ and $ j \in N_{sq}$. This implies the weight distribution listed in Table \ref{evenc3}.

 $ (4) $ When $l_f= 2 $ and $ p \equiv 5 \pmod 8 $, we get 	
 \begin{align*}
 \texttt{wt}(\textbf{c}(a,b))
 =\begin{cases}
 (p-1)p^{2m-2},  & \mathcal{N}_{f}(0)\mathcal{N}_{g}(0)-1  \textup{ times}, \\
 (p-1) p^{2m-2}+(p-5)\varepsilon_f \varepsilon_g \sqrt{p}^{\tau-2}, &   F_5 \textup{ times}, \\
 (p-1)\bk{p^{2m-2}+ \varepsilon_f \varepsilon_g \sqrt{p}^{\tau-2}} , &   F_6 \textup{ times},  
 \end{cases}	
 \end{align*} 
 where we write 
  \begin{align*}		
  F_5 & = \#\{ (a,  b) \in  \mathcal{S}_f \times \mathcal{S}_g :  \frac{g^{\star}(b)}{f^{\star}(a)} \in C_2^{(4,p)}\} \\
  &=\frac{(p-1)^2}{8}(\mathcal{N}_{f}(i)\mathcal{N}_{g}(i)+\mathcal{N}_{f}(j)\mathcal{N}_{g}(j))=\frac{1}{2} F_3,   \\
  F_6 & = p^{\gamma }-\mathcal{N}_{f}(0)\mathcal{N}_{g}(0)-\frac{1}{2} F_3,
  \end{align*} 
  for $ i \in S_q$ and $ j \in N_{sq}$. Thus the result in Table \ref{evenc4} is derived. 	
 \end{proof}

 	\begin{table}[htbp]
 	\centering
 	\caption[short text] {The weight distribution of $C_{D_{f,g}}$ in Theorem \ref{th:s+teven1} when $ l_f=\frac{p-1}{2}$.}
 		\label{evenc1}
 		\begin{tabular}{ll}
 			\hline
 			weight & frequency \\ \hline
 				$ 0 $ & $1 $ \\
 			$ (p-1)p^{2m-2}$  & $p^{\gamma-1}+(p-1) \varepsilon_f \varepsilon_g\sqrt{p}^{\gamma-2}-1 $ \\
 			$ (p-1)\bk{p^{2m-2}+\varepsilon_f \varepsilon_g \sqrt{p}^{\tau-2}} $ & $(p-1)(p^{\gamma-1}- \varepsilon_f \varepsilon_g\sqrt{p}^{\gamma-2})$\\
 			$(p-1)\bk{p^{2m-2}+ (p-1)\varepsilon_f \varepsilon_g  \sqrt{p}^{\tau-4} }$  & $ p^{2m}-p^{\gamma} $ \\
 			\hline
 		\end{tabular}
 	\end{table}
 	\begin{table}[htbp]
 	\centering
 	\caption[short text] {The weight distribution of $C_{D_{f,g}}$ in Theorem \ref{th:s+teven1} when $ l_f=p-1$. }\label{evenc2}
 		\begin{tabular}{ll}
 			\hline
 			weight & frequency \\ \hline
 				$ 0 $ & $1 $ \\
 			$ (p-1)p^{2m-2}$  & $\mathcal{N}_{f}(0)\mathcal{N}_{g}(0)-1 $ \\
 			$ (p-1)\bk{p^{2m-2}+\frac{1}{2}\varepsilon_f \varepsilon_g \sqrt{p}^{\tau-2}} $ & $F_1$\\
 			$ (p-1)\bk{p^{2m-2}+ \varepsilon_f \varepsilon_g \sqrt{p}^{\tau-2}} $ & $ p^{\gamma }-\mathcal{N}_{f}(0)\mathcal{N}_{g}(0)-F_1  $\\
 			$(p-1)\bk{p^{2m-2}+ (p-1)\varepsilon_f \varepsilon_g  \sqrt{p}^{\tau-4} }$  & $ p^{2m}-p^{\gamma} $ \\
 			\hline
 		\end{tabular}
 	\end{table}
 	
 \begin{table}[htbp]
 \centering
 \caption[short text] {The weight distribution of $C_{D_{f,g}}$ in Theorem \ref{th:s+teven1} when $l_f= 2 $ and $ p \equiv 1 \pmod 8 $. }
 \label{evenc3}
 \begin{tabular}{ll}
 \hline
 weight & frequency \\ \hline
 	$ 0 $ & $ 1 $ \\
 	$ (p-1)p^{2m-2}$  & $\mathcal{N}_{f}(0)\mathcal{N}_{g}(0)-1 $ \\
 	$ (p-1) p^{2m-2}+(p-3)\varepsilon_f \varepsilon_g \sqrt{p}^{\tau-2} $ & $ F_3 $\\
 	$ (p-1)\bk{p^{2m-2}+ \varepsilon_f \varepsilon_g \sqrt{p}^{\tau-2}} $ & $ p^{\gamma }-\mathcal{N}_{f}(0)\mathcal{N}_{g}(0)-F_3$\\
 	$(p-1)\bk{p^{2m-2}+ (p-1)\varepsilon_f \varepsilon_g  \sqrt{p}^{\tau-4} }$  & $ p^{2m}-p^{\gamma} $ \\ 	
 	\hline
 	\end{tabular}
 \end{table}

 \begin{table}[htbp]
 \centering
 \caption[short text] {The weight distribution of $C_{D_{f,g}}$ in Theorem \ref{th:s+teven1} when $ l_f=2$ and $ p \equiv 5 \pmod 8 $. }
 	\label{evenc4}
 	\begin{tabular}{ll}
 	\hline
 	weight & frequency \\ \hline
 	$ 0 $ & $ 1 $ \\
 	$ (p-1)p^{2m-2}$  & $\mathcal{N}_{f}(0)\mathcal{N}_{g}(0)-1 $ \\
 	$ (p-1) p^{2m-2}+(p-5)\varepsilon_f \varepsilon_g \sqrt{p}^{\tau-2} $ & $ \frac{1}{2} F_3 $\\
 	$ (p-1)\bk{p^{2m-2}+ \varepsilon_f \varepsilon_g \sqrt{p}^{\tau-2}} $ & $ p^{\gamma }-\mathcal{N}_{f}(0)\mathcal{N}_{g}(0)-\frac{1}{2} F_3 $\\
 	$(p-1)\bk{p^{2m-2}+ (p-1)\varepsilon_f \varepsilon_g  \sqrt{p}^{\tau-4} }$  & $ p^{2m}-p^{\gamma} $ \\ 
 		\hline
 	\end{tabular}
 \end{table}

 \begin{theorem}\label{th:s+tevenB}
 Suppose that $f,\,g \in \rm{WRPB}$ with $ l_g= \frac{p-1}{2} $. Let $s+t$ be even.  
If $l_f=\frac{p-1}{2} $, then $C_{D_{f,g}}$ is a three-weight $[p^{2m-1}-1,2m]$ linear code with its weight distribution listed in Table \ref{evenB1}. If $ l_f =p-1$, then $C_{D_{f,g}}$ is a four-weight $[p^{2m-1}-1,2m]$ linear code with its weight distribution listed in Table \ref{evenB2}. Otherwise if $ l_f =2 $, then $C_{D_{f,g}}$ is a four-weight $[p^{2m-1}-1,2m]$ linear code with its weight distribution listed in Table \ref{evenB3} when $ p \equiv 1 \pmod 8 $, and in Table \ref{evenB4} when $ p \equiv 5 \pmod 8 $. 
 \end{theorem}
 \begin{proof}
 Notation that $   (0,0)  $	is not in $ \mathcal{S}_f \times \mathcal{S}_g  $ since $f,\,g \in \rm{WRPB}$. This theorem can be derived in the same way as Theorem \ref{th:s+teven1} by using Lemmas \ref{lem:car}, \ref{lem:N_f}, \ref{lem:T}, \ref{lem:length} and \ref{lem:N0_e2}. We omit the details here.
 \end{proof}

 \begin{table}[htbp]
 \centering
 \caption[short text] {The weight distribution of $C_{D_{f,g}}$ in Theorem \ref{th:s+tevenB} when $ l_f= \frac{p-1}{2}$. }
 	\label{evenB1}
 	\begin{tabular}{ll}
 		\hline
 		weight & frequency \\ \hline
 			$ 0 $ & $1 $ \\
$ (p-1)\bk{p^{2m-2}-(p-1) \varepsilon_f \varepsilon_g \sqrt{p}^{\tau-4}} $  
& $p^{\gamma-1}+(p-1) \varepsilon_f \varepsilon_g\sqrt{p}^{\gamma-2} $ \\
 $ (p-1)\bk{p^{2m-2}+\varepsilon_f \varepsilon_g \sqrt{p}^{\tau-4}} $
  & $ (p-1)\bk{p^{\gamma-1}-  \varepsilon_f \varepsilon_g \sqrt{p}^{\gamma-2}} $\\
 	$(p-1)p^{2m-2}$  & $ p^{2m}-p^{\gamma}-1 $ \\
 		\hline
 	\end{tabular}
 \end{table}
 
 \begin{table}[htbp]
 	\centering
 	\caption[short text] {The weight distribution of $C_{D_{f,g}}$ in Theorem \ref{th:s+tevenB} when $ l_f= p-1$. }
 	\label{evenB2}
 	\begin{tabular}{ll}
 		\hline
 		weight & frequency \\ \hline
 		$ 0 $ & $1 $ \\
 		$ (p-1)\bk{p^{2m-2}-(p-1) \varepsilon_f \varepsilon_g \sqrt{p}^{\tau-4}} $  
 		& $\mathcal{N}_{f}(0)\mathcal{N}_{g}(0) $ \\
 		$ (p-1)\bk{p^{2m-2}-\frac{p-2}{2} \varepsilon_f \varepsilon_g \sqrt{p}^{\tau-4}} $  
 		& $F_1 $ \\
 		$ (p-1)\bk{p^{2m-2}+\varepsilon_f \varepsilon_g \sqrt{p}^{\tau-4}} $
 		& $ p^{\gamma}-\mathcal{N}_{f}(0)\mathcal{N}_{g}(0)-F_1  $\\
 		$(p-1)p^{2m-2}$  & $ p^{2m}-p^{\gamma}-1 $ \\
 		\hline
 	\end{tabular}
 \end{table}  

 \begin{table}[htbp]
 	\centering
 	\caption[short text] {The weight distribution of $C_{D_{f,g}}$ in Theorem \ref{th:s+tevenB} when $ l_f=2$ and $ p \equiv 1 \pmod 8 $. }
 	\label{evenB3}
 	\begin{tabular}{ll}
 		\hline
 		weight & frequency \\ \hline
 		$ 0 $ & $1 $ \\
 		$ (p-1)\bk{p^{2m-2}-(p-1) \varepsilon_f \varepsilon_g \sqrt{p}^{\tau-4}} $  
 		& $\mathcal{N}_{f}(0)\mathcal{N}_{g}(0) $ \\
 		$ (p-1)p^{2m-2}-(p+1) \varepsilon_f \varepsilon_g \sqrt{p}^{\tau-4} $  
 		& $F_3 $ \\
 		$ (p-1)\bk{p^{2m-2}+\varepsilon_f \varepsilon_g \sqrt{p}^{\tau-4}} $
 		& $ p^{\gamma}-\mathcal{N}_{f}(0)\mathcal{N}_{g}(0)-F_3 $\\
 		$(p-1)p^{2m-2}$  & $ p^{2m}-p^{\gamma}-1 $ \\
 		\hline
 	\end{tabular}
 \end{table}
 
  \begin{table}[htbp]
  	\centering
  	\caption[short text] {The weight distribution of $C_{D_{f,g}}$ in Theorem \ref{th:s+tevenB} when $ l_f= 2$ and $ p \equiv 5 \pmod 8 $. }
  	\label{evenB4}
  	\begin{tabular}{ll}
  		\hline
  		weight & frequency \\ \hline
  		$ 0 $ & $1 $ \\
 		$ (p-1)\bk{p^{2m-2}-(p-1) \varepsilon_f \varepsilon_g \sqrt{p}^{\tau-4}} $  
 		& $\mathcal{N}_{f}(0)\mathcal{N}_{g}(0) $ \\
 		$ (p-1)p^{2m-2}-(3p+1) \varepsilon_f \varepsilon_g \sqrt{p}^{\tau-4} $  
 		& $ \frac{1}{2} F_3 $ \\
 		$ (p-1)\bk{p^{2m-2}+\varepsilon_f \varepsilon_g \sqrt{p}^{\tau-4}} $
 		& $ p^{\gamma}-\mathcal{N}_{f}(0)\mathcal{N}_{g}(0)-\frac{1}{2} F_3 $\\
 		$(p-1)p^{2m-2}$  & $ p^{2m}-p^{\gamma}-1 $ \\
  		\hline
  	\end{tabular}
  \end{table}

\begin{remark}
	In Theorems \ref{th:s+todd}, \ref{th:s+teven1} and \ref{th:s+tevenB}, we dealt with the code $C_{D_{f,g}}$ for $f,g  \in \rm{WRP}$ or $f,g  \in \rm{WRPB}$ with $l_f \in \{2,\frac{p-1}{2},p-1\}$ and $ l_g= \frac{p-1}{2} $, where $ p\equiv 1 \pmod 4 $. 
	However, some of the results coincide with the known ones in the literature. Specifically, when $f,g  \in \rm{WRP}$, the weight distributions in Tables \ref{odd1}, \ref{evenc1} and \ref{evenc3} coincide with the results of Tables 3, 4 and 5 in \cite{Sinak2022}, respectively. If we set $ t=s $ in Tables \ref{evenc1} and \ref{evenc3}, then we get the results of Theorem 4 in \cite{cao2021}.
	When $f,g  \in \rm{WRPB}$, the weight distributions in Tables \ref{evenB1} and \ref{evenB3} coincide with the results of Tables 6 and 7 in \cite{Sinak2022}, respectively. 
\end{remark}
  
 \subsection{The punctured code}
 
 In the following, we study the punctured code from $C_{D_{f,g}}$ by deleting some coordinates of each codeword. As we can see from Tables \ref{odd1}, \ref{odd3}, \ref{evenc1} and \ref{evenB1}, the length and each nonzero Hamming weight have $ p-1 $ as a common divisor. This suggests that they can be punctured into shorter ones. 
 Let $f  \in \rm{WRP}$ or $f  \in \rm{WRPB}$. For any $ x  \in \mathbb{F}_q $, we obtain $ f(x)=0 $ if and only if $ f( z x)=0 $  for all $ z \in \mathbb{F}_p^* $, since $ f( z x)=z^h f(x) $ for an even integer $ h $ with $ \gcd(h-1,p-1)=1 $. Thus we can select a subset $ \overline{D}_{f,g} = \{\overline{(x,y)}:(x,y) \in D_{f,g} \}$ from $ D_{f,g} $ in \eqref{def:D-fg}, such that $ \bigcup_{z \in \mathbb{F}_p^*} z \overline{D}_{f,g} =  D_{f,g}$ forms a partition of $  D_{f,g} $. Hence, we get the punctured code $ C_{\overline{D}_{f,g}} $ from $C_{D_{f,g}}$. Moreover, the code $ C_{\overline{D}_{f,g}} $ is projective since the minimum distance of its dual $ C_{\overline{D}_{f,g}}^\perp $ is at least $ 3 $ as checked in \cite{cao2021}. We can also find some optimal codes when they meet certain specific conditions.

 The following results related to the weight distributions of $ C_{\overline{D}_{f,g}} $ follow directly from Tables \ref{odd1}, \ref{odd3}, \ref{evenc1} and \ref{evenB1}, respectively. Remember that $\tau=2m+s+t$ and $\gamma =2m-s-t$.

\begin{corollary}\label{cor:odd}
Suppose that $f,\,g \in \rm{WRP}$ or $f,\,g \in \rm{WRPB}$ with $ l_g= \frac{p-1}{2} $. Let $s+t$ be odd and $ \bar{d}^\perp $ be the minimum distance of $ C_{\overline{D}_{f,g}}^\perp $. Then $ \bar{d}^\perp \geqslant 3 $. Moreover, if $l_f=\frac{p-1}{2}$, then $C_{\overline{D}_{f,g}}$ is a three-weight $[\frac{p^{2m-1}-1}{p-1},2m]$ linear code with its weight distribution listed in Table \ref{odd1-1}, and if $ l_f =2 $ and $ p\equiv 1 \pmod 8 $, then $C_{\overline{D}_{f,g}}$ is a five-weight $[\frac{p^{2m-1}-1}{p-1},2m]$ linear code with its weight distribution listed in Table \ref{odd3-1}, where $ E_3 $ and $ E_4 $ are computed in Theorem \ref{th:s+todd}. 
\end{corollary}
 	\begin{table}[htbp]
 		\centering
 		\caption[short text]{The weight distribution of $C_{\overline{D}_{f,g}}$ in Corollary \ref{cor:odd} when $ l_f=\frac{p-1}{2}$.}
 		\label{odd1-1}
 		\begin{tabular}{ll}
 			\hline
 			weight & frequency \\ \hline
 			$ 0 $ & $1 $ \\
 			$ p^{2m-2}$  & $p^{2m}-(p-1)p^{\gamma-1}-1 \qquad $ \\
 			$  p^{2m-2}-     \sqrt{p }^{\tau-3}   $ & $\frac{p-1}{2}\bk{p^{\gamma-1}+  \sqrt{p }^{\gamma-1}} $\\
 			$  p^{2m-2}+     \sqrt{p }^{\tau-3}  $  & $ \frac{p-1}{2}\bk{p^{\gamma-1}-  \sqrt{p }^{\gamma-1}} $ \\
 			\hline
 		\end{tabular}
 	\end{table}
 	
 	 	\begin{table}[htbp]
 	 		\centering
 	 		\caption[short text]{The weight distribution of $C_{\overline{D}_{f,g}}$ in Corollary \ref{cor:odd} when $ l_f =2 $ and $ p\equiv 1 \pmod 8 $.}
 	 		\label{odd3-1}
 	 		\begin{tabular}{ll}
 	 			\hline
 	 			weight & frequency \\ \hline
 	 			$ 0 $ & $1 $ \\
 	 			$ p^{2m-2}$  & $p^{2m} -1- E_3 -E_4 -\frac{(p-1)^2}{4} \bk{ \mathcal{N}_{f}(i)\mathcal{N}_{g}(i)+\mathcal{N}_{f}(j)\mathcal{N}_{g}(j)} $ \\
 	 			$ p^{2m-2}-  \varepsilon_f \varepsilon_g \sqrt{p}^{\tau-3}  $
 	 			& $ E_3 $\\
 	 			$ p^{2m-2}+  \varepsilon_f \varepsilon_g \sqrt{p}^{\tau-3} $
 	 			& $ E_4 $\\
 	 			$ p^{2m-2}+2  \varepsilon_f \varepsilon_g \sqrt{p}^{\tau-3}   $  & $ \frac{(p-1)^2}{4}  \mathcal{N}_{f}(i)\mathcal{N}_{g}(i) $ \\
 	 			$  p^{2m-2}-2  \varepsilon_f \varepsilon_g \sqrt{p}^{\tau-3}    $  & $ \frac{(p-1)^2}{4}  \mathcal{N}_{f}(j)\mathcal{N}_{g}(j) $ \\
 	 			\hline
 	 		\end{tabular} 	 		
 	 	\end{table}
 	 	
\begin{corollary}\label{cor:evenU}
  Suppose that $f,\,g \in \rm{WRP}$ with $ l_f=l_g= \frac{p-1}{2} $ and $s+t$ is even. Let $ \bar{d}^\perp $ be the minimum distance of $ C_{\overline{D}_{f,g}}^\perp $. Then $ \bar{d}^\perp \geqslant 3 $ and   $C_{\overline{D}_{f,g}}$ is a three-weight $[\frac{p^{2m-1}-1}{p-1}+ \varepsilon_f \varepsilon_g \sqrt{p}^{\tau-2},2m]$ linear code with its weight distribution listed in Table \ref{evenc1-1}. Moreover, the code $C_{\overline{D}_{f,g}}$ achieves Griesmer bound if $ \tau=4 $ and $ \varepsilon_f \varepsilon_g=-1 $.	
\end{corollary} 
 	\begin{table}[htbp]
 		\centering
 		\caption[short text] {The weight distribution of $C_{\overline{D}_{f,g}}$ in Corollary \ref{cor:evenU}.}
 		\label{evenc1-1}
 		\begin{tabular}{ll}
 		\hline
 		weight & frequency \\ \hline
 		$ 0 $ & $1 $ \\
 		$  p^{2m-2}$  & $p^{\gamma-1}+(p-1) \varepsilon_f \varepsilon_g\sqrt{p}^{\gamma-2}-1 $ \\
 		$  p^{2m-2}+\varepsilon_f \varepsilon_g \sqrt{p}^{\tau-2}  $ & $(p-1)(p^{\gamma-1}- \varepsilon_f \varepsilon_g\sqrt{p}^{\gamma-2})$\\
 		$ p^{2m-2}+ (p-1)\varepsilon_f \varepsilon_g  \sqrt{p}^{\tau-4} $  & $ p^{2m}-p^{\gamma} $ \\
 		\hline
 		\end{tabular}
 	\end{table}
 	
 \begin{corollary}\label{cor:evenB}
 	Suppose that $f,\,g \in \rm{WRPB}$ with $ l_f=l_g= \frac{p-1}{2} $ and $s+t$ is even. Let $ \bar{d}^\perp $ be the minimum distance of $ C_{\overline{D}_{f,g}}^\perp $. Then $ \bar{d}^\perp \geqslant 3 $ and $C_{\overline{D}_{f,g}}$ is a three-weight $[\frac{p^{2m-1}-1}{p-1},2m]$ linear code with its weight distribution listed in Table \ref{evenB1-1}. 
 \end{corollary} 
 \begin{table}[htbp]
 	\centering
 	\caption[short text] {The weight distribution of $C_{\overline{D}_{f,g}}$ in Corollary \ref{cor:evenB}. }
 	\label{evenB1-1}
 	\begin{tabular}{ll}
 		\hline
 		weight & frequency \\ \hline
 		$ 0 $ & $1 $ \\
 		$  p^{2m-2}-(p-1) \varepsilon_f \varepsilon_g \sqrt{p}^{\tau-4}  $  
 		& $p^{\gamma-1}+(p-1) \varepsilon_f \varepsilon_g\sqrt{p}^{\gamma-2} $ \\
 		$  p^{2m-2}+\varepsilon_f \varepsilon_g \sqrt{p}^{\tau-4}  $
 		& $ (p-1)\bk{ p^{\gamma-1}-  \varepsilon_f \varepsilon_g \sqrt{p}^{\gamma-2}} $\\
 		$ p^{2m-2}$  & $ p^{2m}-p^{\gamma}-1 $ \\
 		\hline
 	\end{tabular}
 \end{table}
  
  \begin{example}
  	Let $ f,g: \mathbb{F}_{5^4} \rightarrow \mathbb{F}_5  $ be defined as $f(x)=\textup{Tr}(x^6)$ and $g(y)= \textup{Tr}(y^{26}-y^2)$. Then $ f,\,g \in \rm{WRP} $ with $ s=t=2 $, $ \varepsilon_f=-1$, $ \varepsilon_g=1 $ and $ l_f=l_g=2 $. Their Walsh transforms satisfy $ \widehat{\chi}_f(\alpha) \in \{0,-5^3 \zeta_5^{f^{\star}(\alpha)}\} $ and $ \widehat{\chi}_g(\beta) \in \{0,5^3 \zeta_5^{g^{\star}(\beta)}\} $, where $\alpha, \beta \in \mathbb{F}_{5^4} $ and $ f^{\star}(0)=g^{\star}(0)=0 $. Hence $C_{D_{f,g}}$ is a three-weight code with parameters $[65624,8,50000]$ and the weight enumerator $ 1+ 520 z^{50000}+ 390000 z^{52500}+ 104 z^{62500 } $. 
  	Its punctured code $ C_{\overline{D}_{f,g}} $ has parameters $[16406,8,12500]$ and the weight enumerator $ 1+ 520 z^{12500}+ 390000 z^{13125}+ 104 z^{15625} $.  
  \end{example}
  
  \begin{example}
  	Let $ f,g: \mathbb{F}_{5^3} \rightarrow \mathbb{F}_5  $ be defined as $f(x)=\textup{Tr}(x^{6}+x^2)$ and $g(y)= \textup{Tr}(\theta y^{6}+ \theta^3 y^2)$ for a primitive element $ \theta $ of $ \mathbb{F}_{5^3}^* $. Then $  f,\,g \in \rm{WRP} $ with $ s=0 $, $t=1 $, $ \varepsilon_f=-1 $, $ \varepsilon_g=1 $, $ l_f=l_g=2 $, $ \widehat{\chi}_f(\alpha) \in \{-\sqrt{5}^3 \zeta_5^{f^{\star}(\alpha)}\} $ and $ \widehat{\chi}_g(\beta) \in \{0,5^2 \zeta_5^{g^{\star}(\beta)}\} $, where $ \alpha,\beta \in \mathbb{F}_{5^3} $ and $ f^{\star}(0)=g^{\star}(0)=0 $. Actually, the function $ f $ is quadratic bent and its Walsh transform satisfies $ |\widehat{\chi}_f(\alpha)|^2 =125 $. Hence $C_{D_{f,g}}$ is a three-weight code with parameters $[3124,6,2400]$ and the weight enumerator $ 1+ 1300 z^{2400}+ 13124 z^{2500}+ 1200 z^{2600 } $. 
  	Its punctured code $ C_{\overline{D}_{f,g}} $ has parameters $[781,6,600]$ and the weight enumerator $ 1+ 1300 z^{600}+ 13124 z^{625}+ 1200 z^{650 } $. 
  \end{example}
  
  \begin{example}
  	Let $ f,g: \mathbb{F}_{5^2} \rightarrow \mathbb{F}_5  $ be defined as $f(x)=\textup{Tr}( x^2)$ and $g(y)= \textup{Tr}(\theta y^{2}- \theta  y^6)$ for a primitive element $ \theta $ of $ \mathbb{F}_{5^2}^* $. Then $  f,g  $ are quadratic bent functions in the set $ \rm{WRP} $, with $ s=t=0 $, $ \varepsilon_f=-1 $, $ \varepsilon_g=1 $, $ l_f=l_g=2 $, $ \widehat{\chi}_f(\alpha) \in \{-5 \zeta_5^{f^{\star}(\alpha)} \} $ and $ \widehat{\chi}_g(\beta) \in \{ 5  \zeta_5^{g^{\star}(\beta)} \} $, where $ \alpha,\beta \in \mathbb{F}_{5^2} $ and $ f^{\star}(0)=g^{\star}(0)=0 $. Then the code $C_{D_{f,g}}$ is a two-weight code with parameters $[104,4,80]$ and the weight enumerator $ 1+ 520 z^{80}+ 104 z^{100} $. 
  	Its punctured code $ C_{\overline{D}_{f,g}} $ has parameters $[26,4,20]$ and the weight enumerator $ 1+ 520 z^{20}+ 104 z^{25} $. The punctured code is optimal with respect to the Griesmer bound.
  \end{example}
  
 \section{Minimality of the codes and their applications}
 
 Any linear code can be applied to design secret sharing schemes by considering the access structure. However, the access structure based on a linear code is usually very complicated, and only can be determined exactly in several specific cases. One such case is when the code is minimal.

A linear code $C$ over $ \mathbb{F}_p $ is called minimal if every nonzero codeword $ \textbf{c} $ of $C$ solely covers its scalar multiples $ z\textbf{c} $ for $ z \in \mathbb{F}_p^* $. In 1998, Ashikhmin and Barg \cite{ab1998} provided a well-known criteria for minimal linear codes.
 \begin{lemma}\label{lem11}(Ashikhmin-Barg Bound \cite{ab1998})
 Let $C$ be a linear code over $ \mathbb{F}_p $. Then all nonzero codewords of $C$ are minimal, provided that
 	\begin{equation*}
 	\dfrac{w_{min}}{w_{max}}>\dfrac{p-1}{p}, 
 	\end{equation*}
 	where $w_{min}$ and $w_{max}$ stand for the minimum and maximum nonzero weights in $C$, respectively.
 \end{lemma}
 
 Now we will show under what circumstances the constructed linear codes are minimal according to Lemma \ref{lem11}.
 
 \begin{theorem}
 We have the following bounds on parameters of the code $C_{D_{f,g}}$.
 
$ (1) $	The linear codes described in Tables \ref{odd1},\ref{odd2} and \ref{odd3} are minimal provided when $\varepsilon_f \varepsilon_g \in \{ \pm 1 \} $ and $2m-s-t \geqslant 5$. 
 	
$  	(2) $ The linear codes described in Tables \ref{evenc1}, \ref{evenc2} \ref{evenc3} and \ref{evenc4} are minimal provided when $\varepsilon_f \varepsilon_g = 1 $ and $2m-s-t \geqslant 4$, or $\varepsilon_f \varepsilon_g =  -1$ and $2m-s-t \geqslant 6$.
 	
$ (3) $ The linear codes described in Tables \ref{evenB1}, \ref{evenB2}, \ref{evenB3} and \ref{evenB4} are minimal provided when $\varepsilon_f \varepsilon_g \in \{ \pm 1 \} $ and $2m-s-t \geqslant 4$.	
 \end{theorem}

\begin{remark}
	Our punctured codes $C_{\overline{D}_{f,g}}$ are minimal for almost all cases.
\end{remark}
 
 It should be noticed that the minimum distance of $C_{D_{f,g}}^\perp$ equals $ 2 $ since there are two linearly dependent entries in each codeword in $C_{D_{f,g}}$. So under the framework stated in \cite{ding2003}, the minimal codes described in Theorems \ref{th:s+todd}, \ref{th:s+teven1} and \ref{th:s+tevenB} can be employed to construct high democratic secret sharing schemes with new parameters. The punctured codes are projective and minimal, as we have discussed previously. So they are also suitable for secret sharing schemes. The projective three-weight codes in Tables \ref{odd1-1}, \ref{evenc1-1} and \ref{evenB1-1} can be applied to design association schemes \cite{Calderbank1984}.

\section{Conclusion}
 The paper studied the construction of linear codes using defining set from two weakly regular plateaued functions with index $ \frac{p-1}{2} $ for $ p \equiv 1 \pmod 4 $, and hence, this is an extension of the results in \cite{cao2021}, \cite{Sinak2022} and \cite{wu2020}. The punctured codes were also investigated and we found optimal codes among them. Moreover, our codes are suitable for designing association schemes and secret sharing schemes.


\end{document}